\documentclass[journal]{IEEEtran}

\usepackage{epsfig}
\usepackage{amssymb}
\usepackage{amsthm}
\usepackage{subfigure}
\usepackage{amsmath}

\usepackage{amssymb}
\usepackage{amsfonts}
\usepackage{graphicx}
\usepackage{url}
\usepackage{ccaption}
\usepackage{booktabs}
\usepackage{color}
\usepackage{multirow}

\begin{document}

\title{Bilateral Teleoperation of Multiple Robots under Scheduling Communication}
%
%
%

\author{Yuling~Li,~\IEEEmembership{Member,~IEEE,}
        Kun~Liu,~\IEEEmembership{Member,~IEEE,}
        Wei~He,~\IEEEmembership{Senior Member,~IEEE,}
        Yixin~Yin,~\IEEEmembership{Member,~IEEE,}
        Rolf Johansson,~\IEEEmembership{Fellow,~IEEE,}
         Kai~Zhang

\thanks{Y. Li, Y. Yin, W. He and K. Zhang are with the School of Automation and Electrical Engineering,
University of Science and Technology Beijing, Beijing, 100083, P.~R.~China, and the Key Laboratory of Knowledge Automation for Industrial Processes, Ministry of Education, Beijing 100083, P.~R.~China .\protect\\ E-mails: lyl8ustb@gmail.com, hewei.ac@gmail.com, yyx@ies.ustb.edu.cn, ustb\_zk@126.com.}
\thanks{K. Liu is with the School of Automation, Beijing Institute of Technology, Beijing 100081, P.~R.~China.  \protect\\ E-mail: kunliubit@bit.edu.cn (Corresponding author).}
 \thanks{R. Johansson is with the department of automatic control, Lund University, P.O. Box 118, 22100 Lund, Sweden. \protect E-mail: Rolf.Johansson@control.lth.se}
}

%
%

\markboth{IEEE Transactions on Control Systems Technology}%
{Shell \MakeLowercase{\textit{et al.}}: Bare Demo of IEEEtran.cls for IEEE Journals}
%



\maketitle

\begin{abstract}
In this paper, bilateral teleoperation of multiple slaves coupled to a single master under scheduling communication is investigated. The sampled-data transmission between the master and the multiple slaves is fulfilled over a delayed communication network, and at each sampling instant, only one slave is allowed to transmit its current information to the master side according to some scheduling protocols. To achieve the master-slave synchronization, Round-Robin scheduling protocol and Try-Once-Discard scheduling protocol are employed, respectively. By designing a scheduling-communication-based controller, some sufficient stability criteria related to the controller gain matrices, sampling intervals, and communication delays are obtained for the closed-loop teleoperation system under Round-Robin and Try-Once-Discard scheduling protocols, respectively. Finally, simulation studies are given to validate the effectiveness of the proposed results.
\end{abstract}


\newtheorem{assumption}{Assumption}
\newtheorem{definition}[assumption]{Definition}
\newtheorem{theorem}[assumption]{Theorem}
\newtheorem{remark}[assumption]{Remark}
\newtheorem{proposition}[assumption]{Proposition}
\newtheorem{lemma}[assumption]{Lemma}

%
\IEEEpeerreviewmaketitle

\section{Introduction}
%
%
%
%

\IEEEPARstart{B}{ilateral} teleoperation systems  which allow human operators to extend their intelligence and manipulation skills to remote environments are widely used in applications such as telesurgery, space exploration, nuclear operation,  underwater exploration \cite{Hokayem2006_survey}.  A typical bilateral teleoperation system with a configuration of single-master-single-slave (SMSS) involves two robots which exchange position, velocity and/or haptic information through a networked communication channel.
However, for some complex tasks, teleoperation of one slave robot may fail in completing such tasks where multiple manipulators in cooperation are required. Hence,  teleoperation of multiple slave robots has  emerged to cope with a new set of applications incompatible with SMSS configurations \cite{Ghorbanian2013_dualmaster, Li2010_multiple, Palafox2009_multiple}. Teleoperation of multiple slaves can complete multiple tasks in a shorter time, covering large-scale areas, and with the ability to adapt to single point failures more easily, and hence effectively encompass a broader range of surveillance tasks, military operations, and rescue missions, and so on \cite{rodriguez2010bilateral}.  Teleoperation of multiple slaves can be manipulated by one human operator through one master, or by multiple human operators through multiple masters. In this paper, we restrict our attention to the former one, that is, teleoperation systems with single-master-multiple-slaves (SMMS) configurations.

It should be noted that due to  the distance of teleoperation, communication time delays are inevitable.  The classical approaches to deal with delayed bilateral teleoperation systems are passivity-based approaches \cite{Nuno2011_tutorial}, which are mostly based on scattering theory \cite{Spong1989_classical} and wave variable formalism \cite{Niemeyer1991_classical}. Some other passivity-based controllers relying on damping injection \cite{Nuno2010_improvedsyn}, and adaptive control \cite{ chan2014application, Sarras2014_JFI} were also developed  recently.
 For the passivity-based control methods, the assumption that both of the human operator and the environment be passive was imposed and thus it is restrictive. To remove the passivity assumption about external forces, input-to-state stability/input-to-output stability (ISS/IOS) theory is introduced into the control design and stability analysis of teleoperation systems. By applying ISS/IOS theory, Polushin et. al. \cite{Polushin2003_Iss, polushin2013a} firstly designed PD-based controllers for teleoperation systems, and proved the stability of the closed-loop system with communication delays by constructing two input-to-state stable subsystems. However, in \cite{Polushin2003_Iss, polushin2013a}  the positions of the master and the slave will exactly converge to the origin, which, however, should not be expected in applications of teleoperation systems. In other words, position synchronization between the master and the slave robots are expected, which do not necessarily imply that the positions  should converge to the origin. To overcome this limitation, Zhai et. al. \cite{zhai2016SIIOS} investigated a new IOS framework based on state-independent IOS for nonlinear teleoperation systems with asymmetric time-varying delays, where a switched filter-based control method was developed.
Some other advanced control strategies like  predictive control \cite{Uddin2016survey}, optimal control \cite{li2015guaranteed, mohammadi2017robust}, intelligent control based on fuzzy logic \cite{yangX2014_fuzzy, yang2014fuzzy} or neural networks \cite{wang2017NN, yang2016NN,  he2016constraints}, prescribed-performance-based control \cite{yang2016_finitetime, zhai2017ppc}, etc., have also been developed to deal with other various aspects of teleoperation systems with delays, such as finite-time stability \cite {yang2014_fuzzy, yang2016_finitetime}, guaranteed synchronization performance \cite{li2015guaranteed, yang2016_finitetime, zhai2017ppc, yang2017bimanual}, input saturation \cite{hashemzadeh2013teleoperation, lee2014controller, zhai2016_inputsaturation}, model uncertainties \cite{kim2013a, zhai2015_adaptive}, brain-machine-interface-based teleoperation \cite{qiu2017brain}.
However, most of these results are for SMSS teleoperation systems. Some but very limited ones considered teleoperation systems with SMMS or multiple-masters-multiple-slaves  (MMMS) configurations. Sirouspour \cite{Sirouspour2005} firstly studied the problem of MMMS teleoperation, while communication delays were neglected. Based on two-subsystems-decomposition method, \cite{li2013_NNmultiple} investigated adaptive neural network control for SMMS teleoperation systems with time delays and input dead-zone uncertainties. \cite{zhai2016_adaptive} addressed fuzzy control of MMMS with asymmetric time-varying delays and model uncertainties.

Like many works on networked systems \cite{Siddique2016observer, wen2018scheduling, liu2015networked}, communication bandwidth limitation should not be neglected in system design and synthesis. In many SMMS teleoperation systems with spatially distributed slaves, the output of multiple slaves can not be transmitted to the master  simultaneously because of bandwidth limitation in the communication network. Thus, it is desirable to stipulate that there can be only a limited number of communication slaves to access the network at the same time. Actually, in many practical cases, the communication is orchestrated by a scheduling rule called a protocol, by which the network sources can be properly scheduled. Specifically, in some practical teleoperation systems, only one sensing node (master or slave) is allowed to transmit its data over the communication network at a time, even though there are lots of nodes in the considered systems. However, to the best knowledge of the authors, very few works consider this limitation for teleoperation systems.

Along with another line, one commonly used assumption in teleoperation design is that the data transmission between the master(s) and the slave(s) is continuous in time, which is very restrictive in real applications. The continuous-time information exchanges are quite energy-consuming since the communication channels are always occupied in  high frequencies \cite{Fridman2010Aut, liu2012wirtingers}, and thus, this will increase the design and implementation cost as well. In fact, communications are likely to occur over a digital network in practice, such that the information is exchanged at discrete time intervals. Thus it is desirable to provide new results for  teleoperation systems with discrete-time data transmission.

Motivated by the aforementioned observations, this paper aims to solve the synchronization problem for a class of SMMS teleoperation systems under scheduling communication with discrete-time information exchanges and time-varying communication delays. Thus, the existing continuous-time controller and stability criteria for teleoperation systems are inapplicable for solving the considered problem.  At each sampling instant, only one slave is allowed to transmit its current information to the master side over communication network. The data transmission through the communication channel is discrete, and thus the transmitted signals are kept constant during the sampling period. An explicit expression of the designed controller is given. Only the samples of the position variables of the remote manipulators at discrete time instants are needed, and thus the amount of  transmitted synchronization information greatly reduces and the efficiency of bandwidth usage increases. This makes the tracking of teleoperation systems more efficient and useful in real-life applications. Two kinds of scheduling protocols, i.e., Round-Robin (RR) scheduling protocol and Try-Once-Discard (TOD) scheduling protocol, are provided and employed. By constructing Lyapunov-Krasovskii functionals, some efficient stability criteria in terms of linear matrix inequalities (LMIs) are obtained for the closed-loop SMMS teleoperation systems under scheduling communication network.

The rest of this paper is organized as follows. The problem formulation and some preliminaries are given in Section \ref{sec:pro_formulation}. Controllers for SMMS teleoperation systems with scheduling network and the stability analysis for the closed-loop system under RR and TOD scheduling protocols, respectively, are provided in Section \ref{sec:controller}. Then, some simulation results are given in Section \ref{sec:simulation} for illustration.  Finally, Section \ref{sec:conclusion} concludes the paper.

%
%
%

\textbf{Notations}: Throughout this paper, the superscript $T$ stands for matrix transposition. $\mathbb{R}^n$ denotes the $n$-dimensional Euclidean space with vector norm $|\cdot|$, $\mathbb{R}^{n\times m}$ is the set of all $n\times m$ real matrices. $*$ represents a block matrix which is readily referred by symmetry. $\mathbb{N}$ represents the set of non-negative integers while $\mathbb{N}^{+}$ is the set of  positive integers. $\lambda_{\min}(M)$ and $ \lambda_{\max}(M)$ denote the maximum and the minimum eigenvalue of matrix $M=M^T\in\mathbb{R}^{n\times n}$, respectively. For any function $f:[0,\infty)\rightarrow \mathbb{R}^n$, the $\mathcal{L}_\infty$-norm is defined as $\|f\|_\infty:=\sup_{t\geq 0}|f(t)|$, and the square of the $\mathcal{L}_2$-norm as $\|f\|_2^2:=\int_{0}^\infty |f(t)|^2dt$. The $\mathcal{L}_\infty$ and $\mathcal{L}_2$ spaces are defined as the sets $\{f:[0,\infty)\rightarrow \mathbb{R}^n,\|f\|_\infty<\infty\}$ and $\{f:[0,\infty)\rightarrow \mathbb{R}^n,\|f\|_2<\infty\}$, respectively.


%

\section{Problem Formulation and Preliminaries}\label{sec:pro_formulation}

This paper is concerned with bilateral teleoperation control of multiple  manipulators by a master manipulator over scheduling communication as shown in Fig.~\ref{fig:sche_tele}.
The dynamics of the SMMS teleoperation system consisting of a single $n$-degree of freedom (DOF) master manipulator and $N$ ($N\geq2$) $n$-DOF coordinated slave manipulators  can be described as follows:
\begin{equation}
M_m(q_m)\ddot{q}_m\!+\!C_m(q_m,\dot{q}_m)\dot{q}_m\!+\!G_m(q_m)\!=\!f_m\!+\!\tau_m,\label{eq:master}
\end{equation}

\begin{align}
\left.\begin{aligned}
&M_{s1}\!(\!q_{s1}\!)\!\ddot{q}_{s1}+C_{s1}\!(\!q_{s1},\dot{q}_{s1}\!)\!\dot{q}_{s1}+G_{s1}\!(\!q_{s1}\!)\!=f_{s1}+\tau_{s1},\\
&\quad\quad\quad\quad\quad\quad\quad\quad\quad\quad\quad\quad\vdots\\
&M_{si}(q_{si})\ddot{q}_{si}\!+\!C_{si}(q_{si},\dot{q}_{si})\dot{q}_{si}\!+\!G_{si}(q_{si})\!=\!f_{si}\!+\!\tau_{si},\\
&\quad\quad\quad\quad\quad\quad\quad\quad\quad\quad\quad\quad\vdots\\
&M_{sN}\!(\!q_{sN}\!)\!\ddot{q}_{sN}\!\!+\!\!C_{sN}\!(\!q_{sN}\!,\!\dot{q}_{sN}\!)\!\dot{q}_{sN}\!\!+\!\!G_{sN}\!(\!q_{sN}\!)\!
\!\!=\!\!f_{sN}\!\!+\!\!\tau_{sN}
\end{aligned}\!\!\right\}\!\label{eq:slave}
\end{align}
where $q_m/q_{si},\dot{q}_{m}/\dot{q}_{si},\ddot{q}_m/\ddot{q}_{si}\in\mathbb{R}^n$ are the joint positions, velocities and acceleration measurements of the master/$i$-th slave devices with $i=1, ..., N$, respectively. $M_m, M_{si}$ represent mass matrices, $C_{m}(q_m, \dot{q}_m), C_{si}(q_i, \dot{q}_{si})$ embody Coriolis and centrifugal effects. $\tau_m, \tau_{si}$ are the control forces, and finally $f_m, f_{si}$ are external forces applied to the manipulators. Each robot in (\ref{eq:master}) and (\ref{eq:slave}) possess the structural property of robotic systems, i.e., the following properties \cite{Nuno2011_tutorial}, \cite{Spong2006_book} with $j=m, s1, s2,... sN$, respectively:
\begin{enumerate}
\item[P1.] The inertia matrix $M_j(q_j)$ is a symmetric positive-definite function and is lower and upper bounded. i.e., $0<\lambda_j^m I\leq M_j(q_j)\leq \lambda_j^M I<\infty$, where $\lambda_j^m,\lambda_j^M$ are positive scalars.
\item[P2.] The matrix $\dot{M}_j(q_j)-2C_j(q_j,\dot{q_j})$ is skew symmetric.

\item[P3.] For all $q_j, x,y\in \mathbb{R}^{n}$, there exists a positive scalar $c_i$ such that $|C_i(q_j,x)y|\leq c_j|x||y|.$
\item[P4.] If $\ddot{q}_j$ and $\dot{q}_j$ are bounded, the time derivative of $C_j(q_j,\dot{q}_j)$ is bounded.

\end{enumerate}
\begin{figure}
  \centering
  \includegraphics[width=0.9\linewidth]{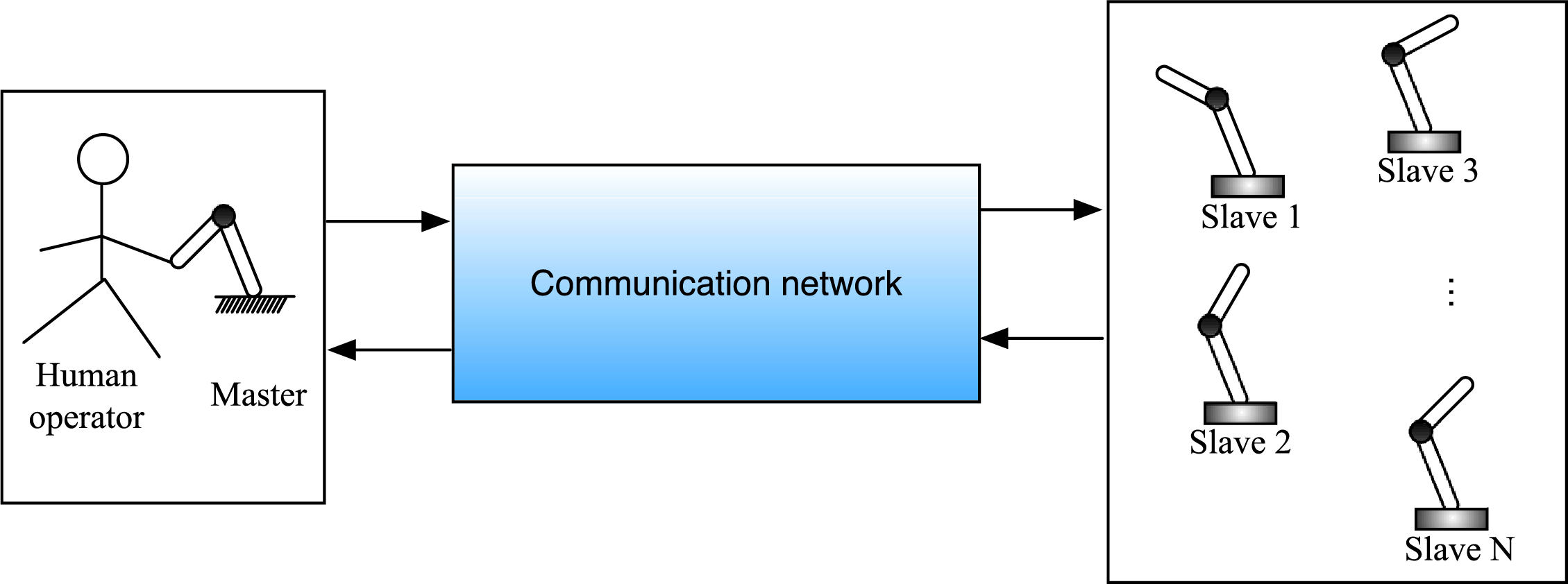}\\
  \caption{Diagram of single-master-multiple-slaves teleoperation system.}\label{fig:sche_tele}
\end{figure}

It is assumed that the data is transmitted from the master to the slaves and from the slaves to the master over delayed communication with  variable and symmetric delays, but only the data of one  manipulator can be transmitted from the local side to the remote side at one time due to the bandwidth limitation of the communication network. Thus the backward communication channel is orchestrated by a scheduling rule called
a protocol. This framework is described in Fig.~\ref{fig:sche_tele1}.

\begin{figure}[h!]
  \centering
  \includegraphics[width=\linewidth]{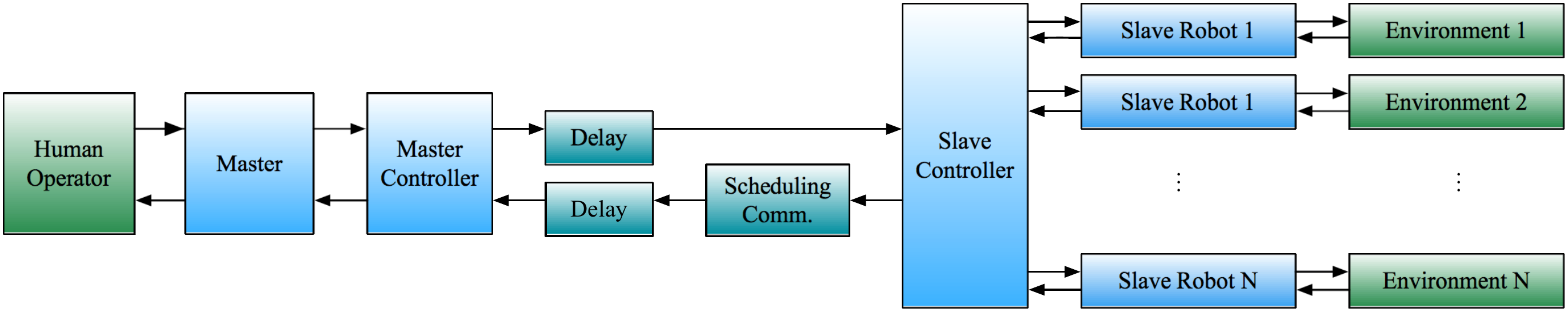}\\
  \caption{Framework of single-master-multiple-slaves teleoperation system with scheduling communication.}\label{fig:sche_tele1}
\end{figure}

 Suppose the sampling at the master and at the slaves are synchronous, and the sampling instants $s_k$, $k=0, 1,2,...$, are a sequence of monotonously increasing constants and satisfy $s_0=0$, $s_k<s_{k+1}$ for $k\in\mathbb{N}$ and $\lim_{k\rightarrow\infty}s_k=+\infty$. At each sampling instant $s_k$, only the sensors at one slave robot are allowed to transmit the sensed information over the communication network according to some scheduling protocols.  Denote by $T_k$ the time needed to transmit the sampling data at the instant $s_k$ to the remote side.  As depicted in Fig.~\ref{fig:sche_tele1}, the master's output is sampled at time instant $s_k$, this sampled information reaches the slave side and updates the Zero-Order Hold (ZOH) at time instant $t_k=s_k+T_k$. Similarly, if the output of one of the slaves is sampled at time instant $s_k$, this information updates the ZOH in the master side at time instant $t_k=s_k+T_k$. The sampling intervals and communication delays are assumed to have certain bounds which are precisely stated in Assumption~\ref{Amp:sampling}.
\begin{assumption}\label{Amp:sampling} There exist positive constants $MATI$ and $MAD$ such that the sampling intervals and communication delays $T_k$ hold for all $k\in\mathbb{N}$:
\begin{enumerate}
\item $s_{k+1}-s_k\leq MATI$,
\item $0\leq T_k\leq MAD$.
\end{enumerate}
\end{assumption}
Note that in Assumption \ref{Amp:sampling} the communication delays are not required to be small with $T_k<s_{k+1}-s_k$. As in \cite{liu2015networked} and \cite{freirich2016decentralized}, we allow  the communication delays to be non-small provided that the old sample cannot get to the remote side after the most recent one. The  time span between the instant $t_{k+1}$ and the current sampling instant $s_k$ is bounded with $t_{k+1}-t_k+T_k\leq MATI+MAD$.

\begin{remark}
Note that in this paper we assume that the signals are transmitted only at each sampling instants, thus the communication delays can be depicted as piecewise-constant functions. This assumption implies discrete-time information exchange between the master and the slaves, which is quite different from the previous works. Discrete-time information exchange would improve the communication efficiency and is energy-saving since the communication channel is not required to be occupied in high frequencies. Furthermore, discrete-time information exchange is more practical in real applications.
\end{remark}
\begin{remark}
Actually, the range of communication delays may vary in an interval with non-zero lower bound  in practice. In this paper, for simplification of analysis, we assume that the lower bound of the communication delays is zero. Our results can be easily extended to the case with non-zero lower bounded delays. Some studies for interval communication delays can be found in \cite{al2017improved, yangX2014_intervaldelay}.
\end{remark}

For the SMMS teleoperation system (\ref{eq:master}-\ref{eq:slave}), we assume that  the slave robots must maintain a distance and orientation from the formation's geometric center at all time. Denote by $\gamma_i$ for the $i$-th robot's distance from the formation's center $\bar{q}_s(t):=(1/N)\sum_{i=1}^Nq_{si}(t)$, we assume $\gamma_i(t)=\gamma_i$ and $\dot{\gamma}(t)=0$ for all $t\geq 0$. Furthermore, we assume that
\[\gamma_i(t)\neq\gamma_j(t),\forall i\neq j, \sum_{i=1}^N\gamma_i(t)=0.\]

This paper aims to offer a stable bilateral control framework that guarantees  master-slaves synchronization under scheduling communication. In summary, two control objectives are provided:
\begin{enumerate}
\item Position synchronization with coordinated motion between the master and the slaves should be achieved: the slave robots follow the master's command while maintaining  a relative distance $\gamma_i$ with respect to the formation's center at all the time, $\bar{q}_s\rightarrow q_m, q_{si}\rightarrow q_{m}+\gamma_i$. In other words, the position errors sastify $e_i:= q_m-(q_{si}-\gamma_i)\rightarrow 0$, $\bar{e}:= q_m-\bar{q}_s\rightarrow 0$.

\item Force tracking between the master and the slaves should be guaranteed, that is, the contribution of environmental forces should be reflected to the operator under steady-state conditions, i.e., $f_m=-\frac{1}{N}\sum_{i=1}^Nf_{si}:=-\bar{f}_s$.

\end{enumerate}

\section{Controller Design and Stability Analysis}\label{sec:controller}
Suppose that the positions  and velocities of the master and the slaves are available for measurement. In this paper, we allow only the positions of the local manipulators to be transmitted to the remote side.
 Since only one slave's position is scheduled to be active to update with the current information $q_{si}(s_k)$ at each sampling instant $s_k$, the updating law of the most recently received position information $\hat{q}_{si}(s_k)$ at the master side  is given as follows:
\begin{equation}\label{eq:update_qsi}
\hat{q}_{si}(s_k)=\left\{
\begin{aligned}
q_{si}(s_k), i=i_k^*\\
\hat{q}_{si}(s_{k-1}), i\neq i_k^*,
\end{aligned}\right.
\end{equation}
where  $i_k^*\in\{1,...,N\}$ is the active scheduled slave at the sampling instant $s_k$ and will be determined by some scheduling protocols provided later.

As described earlier in Section \ref{sec:pro_formulation}, we suppose that the control inputs from the remote side of the teleoperation system are generated by ZOH devices, and the controllers and the ZOH devices update their outputs as soon as they receive the new data, then for $t\in[t_k,t_{k+1})$, the  P+d-like controllers are proposed as follows:
\begin{align}
&\tau_m\!(t\!)\!=\!-\!\frac{1}{N}\![\sum_{i=1}^NK_i^p\!(q_m\!(t\!)\!-\!\hat{q}_{si}\!(s_{k}\!)\!)\!+\!K_i^d\dot{q}_{m}(t)\!]\!+\!G_m\!(q_m\!),\label{eq:taum}\\
&\tau_{si}\!(t\!)\!=\!-K_i^p(\check{q}_{si}(t)-\hat{q}_m(s_k))-K_i^d\dot{q}_{si}(t)+G_{si}(q_{si}),\label{eq:taus}
\end{align}
where $\hat{q}_m(s_k)=q_m(s_k)$, $\hat{q}_{si}$ is provided in (\ref{eq:update_qsi}), $\check{q}_{si}(t)=q_{si}(t)-\gamma_{i}(t)$, $0<K_i^p\in\mathbb{R}^{n\times n}, 0<K_i^d\in\mathbb{R}^{n\times n}$,
 $i\in\{1,..., N\}$. 

 \begin{remark}
 Note that in (\ref{eq:taum}-\ref{eq:taus}) only the position sampling signals need to be transmitted to the remote side through the communication channel. Thus the amount of transmitted synchronization information greatly reduces and the efficiency of bandwidth usage increases, which makes the teleoperation systems more efficient and useful in applications.
 \end{remark}

In the following, the stability analysis of the teleoperation system (\ref{eq:master}-\ref{eq:slave}) under the control (\ref{eq:taum}-\ref{eq:taus}) is provided. Specifically, the teleoperation system under RR and TOD protocols is studied, respectively.
\subsection{RR  protocol}

Under RR scheduling protocol, the measurements of the slaves are transmitted one after another periodically, that is, for each $k\in \mathbb{N}^{+}$, the active slave to access the communication network shall satisfy
\begin{equation}i_k^*=i_{k+N}^* \label{eq:ikstar_RR}\end{equation}
while repeating of the active node index within a circulation is prohibited, that is
\begin{equation}i_p^*\neq i_q^*, ~0< p<q\le N. \label{eq:ik_RR} \end{equation}
Thus, for $i\in\{1, 2, ..., N\}$, the master side $\hat{q}_{si}$  is represented as
\[\hat{q}_{si}(s_k)=q_{si}(s_{k-j}), j\in\{0,1,...,N-1\}.\]
According to the time-delay approach \cite{Fridman2010Aut}, denote $d_{si}(t)=t-s_{k-j}, t\in[t_k,t_{k+1})$ for an arbitrary given $k\in\mathbb{N}$, we have
\begin{align*}
0&\leq d_{si}(t)\leq t_{k+1}-s_{k-j}\\
&=s_{k+1}+T_{k+1}-s_{k-j}\\
&\leq N\cdot MATI+MAD\triangleq h_S.
\end{align*}
 Therefore, when all the measurements are transmitted at least once, i.e., for $t\geq t_{N-1}$, the master controller (\ref{eq:taum}) under RR protocol (\ref{eq:ikstar_RR}-\ref{eq:ik_RR}) can be represented as
 \begin{eqnarray}
\tau_m(t)&=&-\frac{1}{N}\sum_{i=1}^NK_{i}^p(q_m(t)-q_{si}(t-d_{si}(t)))\nonumber\\&&-\frac{1}{N}\sum_{i=1}^NK_i^d\dot{q}_m(t)+G_m(q_m(t)),t\geq t_{N-1}.\label{eq:taum_RR}
 \end{eqnarray}

Similarly, the slave controller (\ref{eq:taus}) under RR protocol (\ref{eq:ikstar_RR}-\ref{eq:ik_RR}) can be rewritten by denoting $d_m(t)=t-s_k$:
\begin{eqnarray}
\tau_{si}(t)&=&-K_i^p(\check{q}_{si}(t)-q_m(t-d_m(t)))-K_i^d\dot{q}_{si}(t)\nonumber\\&&+G_{si}(q_{si}(t)), t\geq t_{N-1}, i=1,..., N, \label{eq:tausi_RR}
\end{eqnarray}
where $0\leq d_m(t)\leq MATI+MAD\triangleq h_M<h_S.$

Substituting  (\ref{eq:taum_RR}-\ref{eq:tausi_RR}) into (\ref{eq:master}-\ref{eq:slave}), we obtain the closed-loop teleoperation system with the following dynamics for $t\geq t_{N-1}$:

\begin{equation}\left\{
\begin{aligned}
f_m(t)\!=\!&M_m(q_m(t))\ddot{q}_m(t)\!+\!C_m(q_m(t),\dot{q}_m(t))\dot{q}_m(t)\\
\!+\!&\frac{1}{N}\sum_{i\!=\!1}^N[K_i^p\!(q_m\!(t\!)\!-\!q_{si}\!(t\!-\!d_{si}\!(t\!)\!)\!)\!+\!K_i^d\dot{q}_m\!(t\!)],\\
f_{si}(t)\!=\!&M_{si}(q_{si}(t))\ddot{q}_{si}(t)\!+\!C_{si}(q_{si}(t),\dot{q}_{si}(t))\dot{q}_{si}(t)\\\!+\!&K_i^d\dot{q}_{si}(t)\!+\!K_i^p(\check{q}_{si}(t)\!-\!q_m(t\!-\!d_m(t))),
\end{aligned}
\right.\\\label{eq:clo_system_RR}\end{equation}
where $d_{si}\in[0,h_S]$, $d_m\in[0,h_M]$. The initial condition for ($\ref{eq:clo_system_RR}$) has the form of $x(t)=\text{col}\{q_m,\dot{q}_m,\check{q}_{s1},\dot{q}_{s1},...,\check{q}_{sN},\dot{q}_{sN}\}=\varphi(t)$, $t\in[-h_S,0]$, $\varphi(0)=x_0$.
The following theorem summaries the main results for the stability of the closed-loop system (\ref{eq:clo_system_RR}) under RR scheduling protocol (\ref{eq:ikstar_RR}-\ref{eq:ik_RR}).
 \begin{theorem}\label{Thm:RR}
 Consider the closed-loop teleoperation system (\ref{eq:clo_system_RR}) with the RR scheduling  protocol (\ref{eq:ikstar_RR}-\ref{eq:ik_RR}).
 If there exist $0<R_m\in\mathbb{R}^{n\times n}$, $0<R_{si}\in\mathbb{R}^{n\times n}$, $i=1,...,N$ such that the following LMIs
 \begin{align}
 \Pi_i\!=\!\begin{bmatrix}
 \Pi_{i}^{1}&0&0&-K_i^p\\
 * &\Pi_{i}^{2}&-K_i^p&0\\
 * & * & -h_M^{-1}R_m&0\\
 * & * & * & -h_S^{-1}R_{si}
 \end{bmatrix}
 \!<\!0\label{eq:LMI_RR}
 \end{align}
with $\Pi_{i}^{1}= -2K_i^d+h_MR_m, \Pi_{i}^{2}=-2K_i^d+h_SR_{si}$
 are satisfied,
 then the following claims hold:\begin{enumerate}
\item if the SMMS teleoperation system (\ref{eq:master}-\ref{eq:slave}) is in free motion, that is, $f_m(t)=f_{si}(t)\equiv 0, i=1,2,.., N$, then all the signals are bounded for $t\geq t_{N-1}$ and the position coordination errors, and velocities asymptotically converge to zero, that is, $\lim_{t\rightarrow \infty} q_m(t)-\check{q}_{si}(t)=\lim_{t\rightarrow \infty} \dot{q}_m(t)=\lim_{t\rightarrow \infty} \dot{q}_{si}(t)=0$, $i=1, 2,..., N$, which implies that $q_m\rightarrow \bar{q}_s$ and $q_{si}\rightarrow q_m+\gamma_i$ as $t\rightarrow \infty$. \label{it:pro_item1}
\item if $f_m\in \mathcal{L}_\infty,f_{si}\in\mathcal{L}_\infty$, then $\dot{q}_m(t), \dot{q}_{si}(t) \in\mathcal{L}_\infty$ for all $t\geq t_{N-1}$.
\item if $f_m\in \mathcal{L}_2,f_{si}\in\mathcal{L}_2$, then all the signals are bounded for all $t\geq t_{N-1}$, and $\lim_{t\rightarrow \infty} q_m(t)-\check{q}_{si}(t)=\lim_{t\rightarrow \infty} \dot{q}_m(t)=\lim_{t\rightarrow \infty} \dot{q}_{si}(t)=0$, which implies that $q_m\rightarrow \bar{q}_s$ and $q_{si}\rightarrow q_m+\gamma_i$ as $t\rightarrow \infty$.
\item the force tracking is guaranteed as the teleoperation system is in steady-state, i.e., $f_m=-\bar{f}_s$.
\end{enumerate}
 \end{theorem}
\begin{proof}
 To develop the stability condition for the closed-loop system (\ref{eq:clo_system_RR}) under RR protocol (\ref{eq:ikstar_RR}-\ref{eq:ik_RR}), we use the following Lyapunov-Krasovskii functional for $t\geq t_{N-1}$:
 \begin{equation}V(t)=V_1(t)+V_2(t)+V_3(t),\label{eq:V_RR}\end{equation}
 where
 \begin{align*}
 V_1(t)&=N\dot{q}_m^T(t)M_m(q_m(t))\dot{q}_m(t)\\
 &\quad+\sum_{i=1}^N\dot{q}_{si}^T(t)M_{si}(q_{si}(t))\dot{q}_{si}(t),\\
  V_2(t)&=\sum_{i=1}^N(q_m(t)-\check{q}_{si}(t))^TK_i^p(q_m(t)-\check{q}_{si}(t)),\\
 V_3(t)&=N\int_{-h_M}^{0}\int_{t+\theta}^t\dot{q}_{m}^T(\delta)R_m\dot{q}_{m}(\delta)d\delta d\theta
 \\&\quad+\sum_{i=1}^N\int_{-h_S}^{0}\int_{t+\theta}^t \dot{q}_{si}^T(\delta)R_{si}\dot{q}_{si}(\delta)d\delta d\theta.
 \end{align*}

The derivative of $V$ along with the trajectory of system (\ref{eq:clo_system_RR}) for $t\geq t_{N-1}$ is
$\dot{V}(t)=\sum_{i=1}^3\dot{V}_i(t)$ with
 \begin{align*}
 \dot{V}_1(t)&=2N\dot{q}_m^T(t)f_m(t)-2\sum_{i=1}^N\dot{q}_m^T(t)K_i^d\dot{q}_m(t)\\&\quad+\sum_{i=1}^N[2\dot{q}_{si}^T(t)f_{si}(t)-2\dot{q}_{si}^T(t)K_i^d\dot{q}_{si}(t)]
\\
 &\quad-2\sum_{i=1}^N\dot{q}_m^TK_i^p(q_m(t)-\check{q}_{si}(t-d_{si}(t)))
 \\&\quad-2\sum_{i=1}^N\dot{q}_{si}^TK_i^p(\check{q}_{si}(t)-q_m(t-d_m(t))),\\
\dot{V}_2(t)&=2\sum_{i=1}^N(q_m(t)-\check{q}_{si}(t))^TK_i^p(\dot{q}_m(t)-\dot{\check{q}}_{si}(t))\\
&=2\sum_{i=1}^N(q_m(t)-\check{q}_{si}(t))^TK_i^p(\dot{q}_m(t)-\dot{q}_{si}(t)),\\
  \dot{V}_3(t)
& \leq Nh_M\dot{q}_m^TR_m\dot{q}_m+\sum_{i=1}^Nh_S\dot{q}_{si}^TR_{si}\dot{q}_{si}
  \\&\quad-Nh_M^{-1}\int_{t-d_m(t)}^t\dot{q}_m^T(\delta)d\delta R_m\int_{t-d_m(t)}^t\dot{q}_m(\delta)d\delta
  \\&\quad-h_{S}^{-1}\sum_{i=1}^N\int_{t-d_{si}(t)}^t\dot{q}_{si}^T(\delta)d\delta R_{si}\int_{t-d_{si}\!(\!t\!)\!}^t\dot{q}_{si}(\delta)d\delta.
 \end{align*}
 Note that \begin{eqnarray}q_m\!(t\!)\!-\!\check{q}_{si}\!(\!t\!-\!d_{si}\!(t\!)\!)\!=\!q_m(t)\!-\!\check{q}_{si}(t)\!+\!\int_{t-d_{si}(t)}^t\dot{q}_{si}(\delta)d\delta\label{eq:positionerror_transform1}\end{eqnarray}
and
\begin{eqnarray}\check{q}_{si}\!(t\!)\!-\!q_m(t\!-\!d_m\!(t\!))=\check{q}_{si}(t)\!-\!q_m(t)\!+\!\int_{t\!-\!d_m(t)}^t\dot{q}_m(\delta)d\delta,\label{eq:positionerror_transform2}\end{eqnarray}
and
 hence one has
 \begin{equation}\label{eq:dotV_RR}
 \dot{V}\!(t)\!\leq\! 2N\dot{q}_m\!^T(t)f_m(t)\!+\!\sum_{i=1}^N\![2\dot{q}_{si}\!^T\!(\!t\!)f_{si}\!(t\!)\!+\!\chi_i\!(\!t\!)\!^T
\Pi_i\chi_i\!(t\!)],
 \end{equation}
 where $\Pi_i$ is given in (\ref{eq:LMI_RR})
 and
 \begin{equation}\chi_i\!=\!\text{col}\{\dot{q}_m,\dot{q}_{si},\int_{t-d_{m}(t)}^t\dot{q}_m(s)ds,\int_{t-d_{si}(t)}^t\dot{q}_{si}(s)ds\}.\label{eq: Xi}\end{equation}

 We first consider the case that $f_m(t)=f_{si}(t)\equiv 0$ for $i=1, ..., N$.
The LMIs (\ref{eq:LMI_RR}) yield
 $\dot{V}(t)\leq 0, t\geq t_{N-1}$, which implies that for $t\geq t_{N-1}$, $V(t)\leq V(t_{N-1})$.
 Furthermore, by (\ref{eq:dotV_RR}), we have
 \begin{equation}\label{eq:dV_simple_RR_noinput}
 \dot{V}(t)\!\leq\!\sum_{i=1}^N\chi_i(t)^T\Pi_i \chi_i(t)\leq \!-\!\sum_{i=1}^N\delta_i|\chi_i(t)|^2, t\geq t_{N-1}
 \end{equation}
 where $\delta_i=-\lambda_{\max}(\Pi_i)$, then integrating
 both sides of (\ref{eq:dV_simple_RR_noinput}) from $t_{N-1}$ to $t$, we have
\begin{eqnarray}\label{eq:V_ditui_RR}
V(t)\leq V(t_{N-1})-\sum_{i=1}^N\delta_i\int_{t_{N-1}}^t|\chi_i(s)|^2ds.
\end{eqnarray}
%
 Clearly, $\chi_i(t)\in\mathcal{L}_2$ and $V(t)\in\mathcal{L}_\infty$ for all $t\geq t_{N-1}$ and thus $\dot{q}_m(t),\dot{q}_{si}(t)\in\mathcal{L}_2$ for all $t\geq t_{N-1}$. The fact that  $V$ is radially unbounded with respect to $q_m-\check{q}_{si}, \dot{q}_m,\dot{q}_{si}$  shows that $q_m(t)-\check{q}_{si}, \dot{q}_m(t),\dot{q}_{si}(t)\in\mathcal{L}_\infty$ for $t\geq t_{N-1}$. Now by (\ref{eq:positionerror_transform1}) and (\ref{eq:positionerror_transform2}), the closed-loop dynamics (\ref{eq:clo_system_RR}) and Properties P1-P4, we know that $\ddot{q}_m(t),\ddot{q}_{si}(t)\in\mathcal{L}_\infty, t\geq t_{N-1}, i=1,...,N$. Thus by Barbalat's Lemma, we have $\lim_{t\rightarrow \infty}\dot{q}_m(t)=\dot{q}_{si}(t)=0$. Furthermore, by (\ref{eq:clo_system_RR}), one has that for any $t\geq t_{N-1}$,

 \begin{align*}
 \dddot{q}_m(t)=&-\frac{dM_m^{-1}(q_m(t))}{dt}[\frac{1}{N}\sum_{i=1}^NK_i^p(q_m(t)-\hat{q}_{si}(s_k)) \\&+\frac{1}{N}\sum_{i=1}^NK_i^d\dot{q}_m(t)
+C_m(q_m(t),\dot{q}_m(t))\dot{q}_m(t)]
\\&-M_m^{-1}(q_m(t))[\dot{C}_m(q_m(t),\dot{q}_m(t))\dot{q}_m(t)\\&+C_m(q_m(t),\dot{q}_m(t))\ddot{q}_m(t)+\frac{1}{N}\sum_{i=1}^NK_i^d\ddot{q}_m(t)
\\&+\frac{1}{N}\sum_{i=1}^NK_i^p\dot{q}_m(t)],
 \\ \dddot{q}_{si}(t)=&-\frac{dM_{si}^{-1}(q_{si}(t))}{dt}[(K_i^p(q_{si}(t)-q_{m}(t-d_m(t)))\\&+K_i^d\dot{q}_{si}(t))+C_{si}(q_{si}(t), \dot{q}_{si}(t))\dot{q}_{si}(t)]
\\&-M_{si}^{-1}(q_{si}(t))[\dot{C}_{si}(q_{si}(t), \dot{q}_{si}(t))\dot{q}_{si}(t)\\&+C_{si}(q_{si}(t),\dot{q}_{si}(t))\ddot{q}_{si}(t)+K_i^d\ddot{q}_{si}(t)+K_i^p\dot{q}_{si}(t)
 ],
 \\i=&1, ..., N,
 \end{align*}
thus $\dddot{q}_m(t),\dddot{q}_{si}(t)\in \mathcal{L}_\infty$ for $t\geq t_{N-1}$ since $\dot{q}_m,\ddot{q}_m,\dot{q}_{si},\ddot{q}_{si}\in\mathcal{L}_\infty$ for $t\geq t_{N-1}$. Thus $\lim_{t\rightarrow\infty}\ddot{q}_m(t)=\ddot{q}_{si}(t)=0$ by Barbalat's Lemma. Now by (\ref{eq:clo_system_RR}), it can be established that $|q_m(t)-\check{q}_{si}(t-d_{si}(t))|, |\check{q}_{si}(t)-q_{m}(t-d_m(t))|\rightarrow 0$ as $t\rightarrow 0$, and thus by (\ref{eq:positionerror_transform1}) and (\ref{eq:positionerror_transform2}), we have $|q_m(t)-\check{q}_{si}(t)|\rightarrow 0$ as $t\rightarrow \infty$.

 If $f_m, f_{si}\in \mathcal{L}_\infty$, that is, there exist positive constants $\Delta_m,\Delta_{si}$ such that $|f_m(t)|\leq \Delta_m, |f_{si}(t)|\leq \Delta_{si}$, then the derivative of the Lyapunov functional (\ref{eq:V_RR}) is given by
 \begin{equation}\dot{V}(t)\!\leq\!\sum_{i=1}^N\!-\!\frac{\delta_i}{2}|\chi_i(t)|^2\!+\!\frac{2}{\delta_i}(|f_m(t)|^2\!+\!|f_{si}(t)|^2), t\!\geq\! t_{N\!-\!1}.\label{eq:dV_simple_RR}\end{equation}

 Integrating
both sides of (\ref{eq:dV_simple_RR}) from $t_{N-1}$ to $t$, we have
\begin{align}
V(t)\leq&\!-\!\sum_{i=1}^N\int_{t_{N-1}}^t\!(\frac{\delta_i}{2} |\chi_i(s)|^2\!-\!\frac{2}{\delta_i}(|f_m(s)|^2\!+\!|f_{si}(s)|^2))ds,
\nonumber\\&+ V(t_{N-1}), t\geq t_{N-1}.\label{eq:dV_ditui_RR}
\end{align}

%
For all $|{\chi _i}(s)| > \frac{2}{{{\delta _i}}}\sqrt {(\Delta _m^2 + \Delta _{si}^2)} ,s \in [{t_{N - 1}},t)$, one has $\sigma_i (|\dot{q}_m(t)|^2+|\dot{q}_{si}(t)|^2+|q_m-\check{q}_{si}|^2)\leq V(t)\leq V(t_{N-1})$, which yields that $\dot{q}_m,\dot{q}_{si}\in\mathcal{L}_\infty, i=1,2,...,N$,
where $\sigma_i=\min\{\lambda_{\min}(NM_m(q_m)),\lambda_{\min}(NM_{si}(q_{si})), \lambda_{\min}(NK_i^p)\}$.  Thus, Claim 2 is established.

 If $f_m,f_{si}\in\mathcal{L}_2$, then (\ref{eq:dV_ditui_RR}) implies that
\begin{align*}&V(t)+\sum_{i=1}^N\int_{t_{N-1}}^t\frac{\delta_i}{2}|\chi_i(s)|^2ds\\\leq&\frac{2}{\delta_i}\sum_{i=1}^N(\|f_m\|_2^2+\|f_{si}\|_2^2)+V(t_{N-1}),t\geq t_{N-1},\end{align*}
 thus $V(t)\in \mathcal{L}_\infty$, $\int_{t_{N-1}}^t|\chi_i(s)|^2ds\in\mathcal{L}_\infty$ for all $t\geq t_{N-1}$, and then $\dot{q}_m,\dot{q}_{si},q_m-\check{q}_{si}\in\mathcal{L}_\infty$ for $t\geq t_{N-1}$. Further with $\int_{t_{N-1}}^t|\dot{q}_m(s)|^2ds\leq \int_{t_{N-1}}^t|\chi_i(s)|^2ds\in\mathcal{L}_\infty$ and $\int_{t_{N-1}}^t|\dot{q}_{si}(s)|^2ds\leq \int_{t_{N-1}}^t|\chi_i(s)|^2ds\in\mathcal{L}_\infty, i=1,...,N$, thus by
 Barbalat's Lemma, we have $\dot{q}_m(t)\rightarrow 0,\dot{q}_{si}(t)\rightarrow 0$ as $t\rightarrow \infty$.
 Furthermore, by
 (\ref{eq:positionerror_transform1}-\ref{eq:positionerror_transform2}), we know that $q_m(t)-q_{si}(t-d_{si}(t)),q_{si}(t)-q_m(t-d_m(t))\in\mathcal{L}_\infty$ for $t\geq t_{N-1}$.
Now following the same line of reasoning in the last part of the proof for Claim 1, the proof is completed using Barbalat's Lemma.

The proof of the last claim is obvious. Applying  $\dot{q}_m=0, \dot{q}_{si}=0$ and $\ddot{q}_m= 0, \ddot{q}_{si}= 0$, $i=1,..., N$ to the closed-loop system dynamics  (\ref{eq:clo_system_RR}),  we have that $f_m(t)=\frac{1}{N}\sum_{i=1}^NK_i^p(q_m(t)-q_{si}(t-d_{si}(t))=\frac{1}{N}\sum_{i=1}^N(-f_{si}(t))=-\bar{f}_{s}(t)$.


 \end{proof}

\subsection{TOD protocol}
In TOD protocol, the sensor node of  the slave robot $i\in\{1,...,N\}$ with the greatest weighted error
\begin{align}&\eta_i(t)=\hat{q}_{si}(s_{k-1})-q_{si}(s_k), i=1,...,N, \nonumber
\\&t\in[t_k,t_{k+1}),k\in\mathbb{N}, \hat{q}(s_{-1})=0,\label{eq:eta_TOD}
\end{align}
will be granted the access to the communication network. Differently from RR scheduling protocol, TOD scheduling protocol does not guarantee that each node will transmit once every $N$ transmissions.

\begin{definition}(TOD Protocol). Let $Q_i>0 (i=1, ..., N)$ be some weighting matrices.  At the sampling instant $s_k$, the weighted TOD protocol is a protocol for which the  measurements of  the slave robot  to be transmitted with the index $i_k^*$ is defined as any index that satisfies
\begin{eqnarray}
&&|\sqrt{Q_{i_k^*}}\eta_{i_k^*}(t)|^2\geq |\sqrt{Q}_i\eta_i(t)|^2,
\nonumber\\&&t\in[t_k,t_{k+1}), k\in\mathbb{N},i=1,...,N.\label{eq:Def_TOD}
\end{eqnarray}
\end{definition}
Here the weighting matrices $Q_{1},..., Q_N$ are variables to be designed. Therefore, the master controller (\ref{eq:taum}) under TOD protocol can be represented as
\begin{align}
\tau_m(t)=&-\frac{1}{N}[K_{i_k^*}^p(q_m(t)-q_{si_k^*}(s_k))+\sum_{i=1}^NK_i^d\dot{q}_m(t)\nonumber\\
&+\sum_{i=1,i\neq i_k^*}^NK_i^p(q_m(t)-\hat{q}_{si}(s_{k-1}))]
 +G_m(q_m(t)),\nonumber\\&t\in[t_k,t_{k+1}), k\in\mathbb{N}.\label{eq:taum_TOD}
\end{align}
Denote \[d(t)=t-s_k, t\in[t_k,t_{k+1}),k\in\mathbb{N},\]
then
\begin{equation}0\leq d(t)\leq MATI+MAD=h_M.\label{eq:D_Bound_TOD}\end{equation}


By time-delay approach, substituting (\ref{eq:taus}) and (\ref{eq:taum_TOD}) into (\ref{eq:master}-\ref{eq:slave}), one has the following dynamics:
\begin{equation}\label{eq:clo_system_TOD}\left\{
\begin{aligned}
&f_m(t)\!=\!
M_m(q_m(t))\ddot{q}_m(t)\!+\!C_m(q_m(t),\dot{q}_m(t))\dot{q}_m(t)
\\&\quad\quad\quad+\frac{1}{N}\sum_{i=1}^NK_i^d\dot{q}_m(t)-\frac{1}{N}\sum_{i=1,i\neq i_k^*}^NK_i^p\eta_i(t)\\&\quad\quad\quad+\frac{1}{N}\sum_{i=1}^NK_i^p(q_m(t)-q_{s_i}(t-d(t)),
\\
&f_{si}(t)\!=\!M_{si}(q_{si}(t))\ddot{q}_{si}(t)\!+\!C_{si}(q_{si}(t),\dot{q}_{si}(t))\dot{q}_{si}(t)\\
&\quad\quad\quad+K_i^d\dot{q}_{si}(t)+K_i^p(\check{q}_{si}(t)-q_m(t-d(t)),
\\
&\dot{\eta}_i(t)\!=\!0, i=1,...,N,t\in[t_k,t_{k+1}), k\in\mathbb{N}.
\end{aligned}
\right.\end{equation}
By (\ref{eq:update_qsi}), we obtain for $i=i_k^*\in\{1,2,...,N\}$,
\begin{equation}
\eta_i(t_{k+1})=\hat{q}_{si}(s_k)-q_{si}(s_{k+1})=q_{si}(s_k)-q_{si}(s_{k+1}),\label{eq:eta_update1_TOD}
\end{equation}
and for $i\neq i_k^*,  i\in\{1,2,...,N\}$
\begin{equation}
\eta_i(t_{k\!+\!1})\!=\!\hat{q}_{si}\!(s_k\!)\!-\!q_{si}\!(s_{k\!+\!1})\!=\!\eta_i(t_k)\!+\!q_{si}(s_k)\!-\!q_{si}(s_{k+1}).\label{eq:eta_update2_TOD}
\end{equation}
Thus, following \cite{freirich2016decentralized} the delayed reset system is given by
\begin{equation}\left\{
\begin{aligned}
&q_{si}(t_{k+1})=q_{si}(t_{k+1}^{-})\\
&\eta_i(t_{k+1})=[1-\delta(i,i_k^*)]\eta_{i}(t_k)+q_{si}(s_k)-q_{si}(s_{k+1})\\
&k\in\mathbb{N}, i={1, 2,..., N},
\end{aligned}
\right.\label{eq:resetcondition}
\end{equation}
where $\delta$ is Kronecker delta.  Summarizing (\ref{eq:clo_system_TOD}-\ref{eq:resetcondition}) we obtain the hybrid model of the closed-loop teleoperation system. The initial condition for (\ref{eq:clo_system_TOD}) and (\ref{eq:resetcondition}) has the form of $x(t)=[q_m^T(t),\dot{q}_m^{T}(t),\check{q}_{s1}^T(t),\dot{q}_{s1}^T(t),...,\check{q}_{sN}^T(t),\dot{q}_{sN}^T(t)]^T=\phi(t)$, $t\in[-h_M,0]$, $\phi(0)=x_0$, and $\eta_i(0)=-q_{si}(0)$, $\phi(t)$ is a continuous function on $[-h_M,0]$.

For the closed-loop model (\ref{eq:clo_system_TOD}) and (\ref{eq:resetcondition}) we now consider the following functional:
\begin{equation}
V_e(t)=V(t)+V_G(t)+\sum_{i=1}^{N}\eta^T_i(t)Q_i\eta_i(t)+W_e(t),\label{eq:Lyp_Ve_TOD}
\end{equation}
where $V(t)$ is defined in (\ref{eq:V_RR}) with $h_S$ replaced by $h_M$
 and
\begin{align}
W_e(t)&=\sum_{i=1,i\neq i_k^*}^N\frac{t_k-t}{t_{k+1}-t_k}\eta^T_i(t)U_i\eta_i(t),\label{eq:We}\\
V_G(t)&=\sum_{i=1}^Nh_M\int_{s_k}^t\dot{q}_{si}^T(s)G_{i}\dot{q}_{si}(s)ds.\label{eq:VG}
\end{align}
Here the term $\eta_i^T(t)Q_i\eta_i(t)\equiv \eta_i^T(t_k)Q_i\eta_i(t_k), t\in[t_k,t_{k+1}), i=1,2,...,N$
is piecewise-constant. $V(t)$ represents a Lyapunov functional for teleoperation systems with delay $d(t) \in[0,h_M]$.  Following  \cite{liu2015networkedSIAM} where an exponential form of $V_G$ has been introduced, the piecewise-continuous in time term $V_G$ in (\ref{eq:VG}) is used to cope with the delays in the reset conditions:
\begin{align}
&V_G(t_{k+1})-V_G(t_{k+1}^{-})
\nonumber\\&=h_M\sum_{i=1}^N\int_{s_{k+1}}^{t_{k+1}}\dot{q}_{si}^T(s)G_{i}\dot{q}_{si}(s)ds
\nonumber\\&\quad-h_M\sum_{i=1}^N\int_{s_{k}}^{t_{k+1}}\dot{q}_{si}^T(s)G_{i}\dot{q}_{si}(s)ds
\nonumber\\&=-h_M\sum_{i=1}^N\int_{s_{k}}^{s_{k+1}}\dot{q}_{si}^T(s)G_{i}\dot{q}_{si}(s)ds
\nonumber\\&\!\leq \!-\!\sum_{i=1}^N[q_{si}(s_{k})\!-\!q_{si}(s_{k\!+\!1})]^T\!G_i\![q_{si}(s_{k})\!-\!q_{si}(s_{k\!+\!1})].\label{eq:VG_jumps}
\end{align}
The $\eta_i$-related term $W_e$ which is non-positive was inspired by  \cite{freirich2016decentralized} to provide non-positive terms of $\eta_i$ in the derivative of $V_e$.
The function $V_e$ is thus continuous and differentiable over $[t_k,t_{k+1})$. The following lemma gives sufficient conditions for the positivity of $V_e$ and for the fact that it does not grow at the jumps $t_k$:

\begin{lemma}\label{lem:TOD}
If there exist $0<Q_i\in\mathbb{R}^{n\times n}$, $0<U_i\in\mathbb{R}^{n\times n}$ and $0<G_i\in \mathbb{R}^{n\times n}$, $i=1,2,...,N$ that satisfy the LMIs
\begin{align}
\Omega_i=\begin{bmatrix}-\frac{1}{N-1}Q_{i}\!+\!U_i&Q_i\\ * &-G_i\!+\!Q_i \end{bmatrix}<0, i=1,2...,N,
\label{eq:Omega_TOD}
\end{align}
then $V_e(t)$ of (\ref{eq:Lyp_Ve_TOD}) is positive in the sense that
\begin{equation}V_e(t)\geq \beta (|\dot{q}_m|^2+|\dot{q}_s|^2+|e|^2)\label{eq:Ve_positivity}\end{equation}
with
$\dot{q}_s\triangleq \text{col}\{\dot{q}_{s1},...,\dot{q}_{sN}\}$,
$e \triangleq \text{col}\{e_{1},...,e_{N}\}$,
$e_i \triangleq q_m-\check{q}_{si}, i=1,..., N,$
for some $\beta>0$. Moreover, $V_e$ does not grow at the jumps along with (\ref{eq:clo_system_TOD}) and (\ref{eq:resetcondition}):
\begin{equation}V_e(t_{k+1})-V_e(t_{k+1}^{-})\leq 0.\label{eq:Ve_ditui}\end{equation}
\end{lemma}

\begin{proof}
It can be seen that (\ref{eq:Omega_TOD}) implies that for $i=1, ..., N$,
\[ U_i<\frac{1}{N-1}Q_i\leq Q_i,\]
thus for all $t\in[t_k,t_{k+1})$ we have
\begin{eqnarray*}
\sum_{i=1}^N\eta_i^T(t)Q_i\eta_i(t)+W_e(t)\geq  \eta_{i_k^*}^TQ_{i_k^*}\eta_{i_k^*}\geq 0,
\end{eqnarray*}
which
yields (\ref{eq:Ve_positivity}).
Now we show that $V_e$ does not grow at the jumps. Since $V(t_{k+1})=V(t_{k+1}^-)$ and $\eta(t_{k+1}^-)=\eta(t_k)$, by taking into account (\ref{eq:VG_jumps}) we obtain \begin{align*}
&V_e(t_{k+1})-V_e(t_{k+1}^{-})\\&\leq\sum_{i=1}^N\eta^T_i(t_{k+1})Q_i\eta_i(t_{k+1})-\sum_{i=1}^N\eta^T_i(t_{k})Q_i\eta_i(t_{k})
\\&\quad\!+\!\sum_{i=1,i\neq {i_k^*}}\eta_i(t_k)^TU_i\eta_i(t_k)+V_G(t_{k+1})-V_G(t_{k+1}^{-})
\\&\leq \eta^T_{i_k^*}(t_{k+1})Q_{i_k^*}\eta_{i_k^*}(t_{k+1})
\\&\quad\!+\sum_{i=1,i\neq i_k^*}^N[\eta^T_i(t_{k+1})Q_i\eta_i(t_{k+1})-\eta_i^T(t_k)(Q_i-U_i)\eta_i(t_k)]
\\&\quad-\eta_{i_k^*}^T(t_k)Q_{i_k^*}\eta^T_{i_k^*}(t_k)\!+V_G(t_{k+1})\!-V_G(t_{k+1}^{-}).
\end{align*}
Note that under TOD protocol
\begin{equation*}
-\eta^T_{i_k^*}(t_k)Q_{i_k^*}\eta_{i_k^*}(t_k)\leq -\frac{1}{N-1}\sum_{i=1,i\neq i_k^*}^N\eta^T_{i}(t_k)Q_{i}\eta_{i_k}(t_k).
\end{equation*}
Denote $\zeta_i(t)=\text{col}\{\eta_i(t),q_{si}(s_{k})-q_{si}(s_{k+1})\}$. By using (\ref{eq:eta_update1_TOD}-\ref{eq:eta_update2_TOD}), we arrive at
\begin{eqnarray*}
V_e(t_{k+1})-V_e(t_{k+1}^{-})&\leq& -\eta^T_{i_k^*}(t_{k+1})(G_{i_k^*}-Q_{i_k^*})\eta_{i_k^*}(t_{k+1})\\&&-\sum_{i=1,i\neq i_k^*}^N\zeta_i^T\Omega_i\zeta_i.
\end{eqnarray*}
This completes the proof.
\end{proof}
By applying Lemma \ref{lem:TOD} and the standard arguments for the stability analysis, we obtain the following result.
\begin{theorem}\label{Thm:TOD}
Consider the hybrid closed-loop system (\ref{eq:clo_system_TOD}) and (\ref{eq:resetcondition}) under TOD scheduling protocol (\ref{eq:Def_TOD}). If there exist  $n\times n$ matrices $ R_m>0, R_{si}>0, G_i>0, Q_i>0, U_i>0$ that satisfy the LMIs (\ref{eq:Omega_TOD}) and

\begin{align}
 &\Sigma_i=
 \begin{bmatrix}
\bar{\Pi}_i&L_i^T\\
* &\Psi_i
 \end{bmatrix}<0,i=1, 2,...,N\label{eq:LMI_TOD}
 \end{align}
with
\begin{align*}
\bar{\Pi}_i&=\Pi_i+h_MD_1^TG_iD_1,\\
L_i&=[1,0,0,0]\otimes \text{col}_{i=1,...,N}\{(K_i^p)^T,j\neq i\},\\
\Psi_i&=-h_M^{-1}\text{diag}_{i=1,...,N}\{U_j,j\neq i\},\\
D_1&=[0_n,I_n,0_n,0_n],
\end{align*}
and $\Pi_i$ given in (\ref{eq:LMI_RR}) with $h_S$ replaced by $h_M$,
then the following claims hold\begin{enumerate}
\item if the SMMS teleoperation system (\ref{eq:master}-\ref{eq:slave}) is in free motion, that is, $f_m(t)=f_{si}(t)\equiv 0$, then all the signals are bounded and the position coordination errors and velocities asymptotically converge to zero, that is, $\lim_{t\rightarrow \infty} q_m(t)-\check{q}_{si}(t)=\lim_{t\rightarrow \infty} \dot{q}_m(t)=\lim_{t\rightarrow \infty} \dot{q}_{si}(t)=0$, which implies that $q_m\rightarrow \bar{q}_s$ and $q_{si}\rightarrow q_m+\gamma_i$ as $t\rightarrow \infty$. \label{it:pro_item1_TOD}
\item if $f_m\in \mathcal{L}_\infty,f_{si}\in\mathcal{L}_\infty$, then $\dot{q}_m(t), \dot{q}_{si}(t)\in\mathcal{L}_\infty$ for all $t\geq 0$.
\item if $f_m\in \mathcal{L}_2,f_{si}\in\mathcal{L}_2$, then all the signals are bounded and $\lim_{t\rightarrow \infty} q_m(t)-\check{q}_{si}(t)=\lim_{t\rightarrow \infty} \dot{q}_m(t)=\lim_{t\rightarrow \infty} \dot{q}_{si}(t)=0$, which implies that $q_m\rightarrow \bar{q}_s$ and $q_{si}\rightarrow q_m+\gamma_i$ as $t\rightarrow \infty$.
\item the force tracking is guaranteed as the teleoperation system is in steady-state, i.e., $f_m=-\bar{f}_s$.
\end{enumerate}

\end{theorem}
\begin{proof}
Choose the Lyapunov-Krasovskii functional (\ref{eq:Lyp_Ve_TOD}). Calculating the derivative of $V_e$ along with the trajectory of system (\ref{eq:clo_system_TOD}) and (\ref{eq:resetcondition}) for $t\in[t_k, t_{k+1})$, $k\in\mathbb{N}$, we have

\begin{eqnarray*}
\dot{V}_e(t)&\leq& \sum_{i=1}^3\dot{V}_i(t)+\sum_{i=1}^Nh_M\dot{q}^T_{si}(t)G_i\dot{q}_{si}(t)\\&&-\sum_{i=1,i\neq i_k^*}\frac{1}{t_{k+1}-t_k}\eta_i^T(t)U_i\eta_i(t)
\\&\leq&\sum_{i=1}^3\dot{V}_i(t)+\sum_{i=1}^Nh_M\dot{q}^T_{si}(t)G_i\dot{q}_{si}(t)\\&&-\sum_{i=1,i\neq i_k^*}\frac{1}{h_M}\eta_i^T(t)U_i\eta_i(t)
\\&=&2N\dot{q}_m^T(t)f_m(t)+\sum_{i=1}^N[2\dot{q}_{si}^T(t)f_{si}(t)\\&&+\chi_i(t)\Pi_i\chi_i(t)+h_M\dot{q}_{si}^T(t)G_i\dot{q}_{si}(t)]
\\&&\sum_{i=1,i\neq i_k^*}[-\frac{1}{h_M}\eta_i^T(t)U_i\eta_i(t)\!+\!2\dot{q}_m^T(t)K_i^p\eta_i(t)],
\end{eqnarray*}
where $\Pi_i$ is given  in (\ref{eq:LMI_RR}) with $h_S$ replaced by $h_M$, and $\chi_i$ is given in (\ref{eq: Xi}) with $d_m(t), d_{si}(t)$ replaced by $d(t)$.
  Define
  \begin{align*}
  \xi_i(t)=&\text{col}\{\chi_i(t),\rho_{i_k^*}(t)\},
  \\=&\text{col}\{\dot{q}_m\!(t\!),\dot{q}_{si}\!(t\!),\int_{t\!-\!d(t)}^t\dot{q}_m\!(s\!)ds,\int_{t\!-\!d(t)}^t\dot{q}_{si}\!(s\!)ds,\rho_{i_k^*}\!(t\!)\},\\
  \rho_{i}(t)=&\text{col}_{i=1,2,...,N}\{\eta_j(t),j\neq i\}.
  \end{align*}
We arrive at
\begin{align}
\dot{V}_e(t)\leq &\sum_{i=1}^N [2\dot{q}_{si}^T(t)f_{si}(t)\!+\!\xi_i^T(t)\Sigma_i\xi_i(t)]
\nonumber\\& +2N\dot{q}_m^T(t)f_m(t), t\in[t_k,t_{k+1}).\label{eq:dotVe_final}
\end{align}

If $f_m(t)=f_{si}(t)\equiv 0$, then the LMI condition (\ref{eq:LMI_TOD}) implies that
\begin{equation}\label{eq:dotV_TOD_noforce}
{\dot V_e}\!(t\!) \!\le\! \sum\limits_{i = 1}^N {\xi _i^T\!(t\!)\Sigma_i {{\xi _i}\!(t\!) \!\le\!  \!-\! \sum\limits_{i = 1}^N {{\mu _i}{{\left| {{\xi _i}\!(t\!)} \right|}^2} \!\le\! 0,t \!\in\! \![{t_k},{t_{k + 1}}\!)} } }
\end{equation}
where $\mu_i=-\lambda_{\max}(\Sigma_i)$.
By (\ref{eq:Ve_ditui}) and (\ref{eq:dotV_TOD_noforce}), we have
\[V_e(t)\leq V_e(t_k)\leq V_e(t_k^-)\leq \cdots\leq V_e(0).\]
Integrating both sides of (\ref{eq:dotV_TOD_noforce}) from $t_{k}$ to $t$, $t\in[t_k,t_{k+1})$,
\begin{equation}
V_e(t)\leq V_e(t_k)-\sum_{i=1}^N\mu_i\int_{t_k}^t|\xi_i(t)|^2ds.\label{eq:V_ditui_TOD}\end{equation}

The inequality (\ref{eq:V_ditui_TOD}) with $t=t_{k+1}^-$ implies that
\begin{eqnarray}V_e(t_{k+1})&\leq& V_e(t_k)-\sum_{i=1}^N\mu_i\int_{t_k}^{t_{k+1}}|\xi_i(s)|^2ds\nonumber\\
&\leq&V_e(t_{k-1})-\sum_{i=1}^N\mu_i\int_{t_{k-1}}^{t_{k+1}}|\xi_i(s)|^2ds\nonumber\\
&\vdots&\nonumber\\
&\leq&V_e(0)-\sum_{i=1}^N\mu_i\int_{0}^{t_{k+1}}|\xi_i(s)|^2ds. \label{eq:Vtk_ditui_TOD}
\end{eqnarray}

 Replacing in (\ref{eq:Vtk_ditui_TOD}) $k+1$ by $k$ and using (\ref{eq:V_ditui_TOD}), we arrive at
 \begin{equation}V_e(t)\leq V_e(0)-\sum_{i=1}^N\mu_i\int_{0}^{t}|\xi_i(s)|^2ds),t\in[t_k,t_{k+1}).\label{eq:Vt_TOD}\end{equation}

Letting $k\rightarrow\infty$ clearly implies that $\xi_i\in\mathcal{L}_2$ and $V_e(t)\in\mathcal{L}_\infty$ for all $t\geq 0$ and thus $\dot{q}_m,\dot{q}_{si}, \int_{t-d(t)}^t\dot{q}_m(s)ds,\int_{t-d(t)}^t\dot{q}_{si}(s)ds, \rho_{i_k^*}\in\mathcal{L}_2$. The fact (\ref{eq:Ve_positivity}) shows that $q_m-\check{q}_{si},\dot{q}_m,\dot{q}_{si}\in\mathcal{L}_\infty$, thus together with (\ref{eq:eta_TOD}), one has that $\eta_i\in\mathcal{L}_\infty\cap\mathcal{L}_2$. Now by
  \begin{align}&q_m(t)-\check{q}_{si}(t-d(t))=q_m(t)-\check{q}_{si}(t)+\int_{t-d(t)}^t\dot{q}_{si}(\delta)d\delta,
  \nonumber\\
 & t\in[t_k,t_{k+1}),k\in\mathbb{N},\label{eq:positionerror_transform1_TOD}\end{align}
and
\begin{align}&\check{q}_{si}(t)-q_m(t-d(t))=\check{q}_{si}(t)-q_m(t)+\int_{t-d(t)}^t\dot{q}_m(\delta)d\delta,
\nonumber\\&t\in[t_k,t_{k+1}),k\in\mathbb{N},\label{eq:positionerror_transform2_TOD}\end{align}
and the closed-loop dynamics (\ref{eq:clo_system_TOD}) and Properties P1-P4, we know that $\ddot{q}_m,\ddot{q}_{si}\in\mathcal{L}_\infty$, thus by Barbalat's Lemma, we have $\lim_{t\rightarrow\infty}\dot{q}_m(t)=\lim_{t\rightarrow \infty}\dot{q}_{si}(t)=0$.

 Furthermore, following the same line of reasoning in the proof of Theorem \ref{Thm:RR}, the closed-loop dynamics (\ref{eq:clo_system_TOD}) implies that $\dddot{q}_m,\dddot{q}_{si}\in \mathcal{L}_\infty$ since $\dot{q}_m,\ddot{q}_m,\dot{q}_{si},\ddot{q}_{si},\eta_i\in\mathcal{L}_\infty$ .  Thus $\lim_{t\rightarrow\infty}\ddot{q}_m(t)=\ddot{q}_{si}(t)=0$ by Barbalat's Lemma. Now by (\ref{eq:clo_system_TOD}), it can be established that $|q_m(t)-\check{q}_{si}(t-d(t))|, |\check{q}_{si}(t)-q_{m}(t-d(t))|\rightarrow 0$ as $t\rightarrow 0$, and thus $|q_m-\check{q}_{si}|\rightarrow 0$ as $t\rightarrow \infty$.


 If $f_m, f_{si}\in \mathcal{L}_\infty$, that is, there exist positive constants $\Delta_m,\Delta_{si}$ such that $|f_m(t)|\leq \Delta_m, |f_{si}(t)|\leq \Delta_{si}$, then the derivative of the Lyapunov functional (\ref{eq:Lyp_Ve_TOD}) is given by
 \begin{equation}\label{eq:dV_simple_TOD}
 {{\dot V}_e}(t) \!\le\! \sum\limits_{i = 1}^N {\!-\! \frac{{{\mu _i}}}{2}} {\left| {{\xi _i}\!(\!t\!)\!} \right|^2} \!+\! \frac{2}{{{\mu _i}}}\!(\!{\left| {{f_m}\!(\!t\!)\!} \right|^2} \!+\! {\left| {{f_{si}}\!(\!t\!)\!} \right|^2}\!),t \!\in\! [{t_k},{t_{\!k\!+\!1\!}})
 \end{equation}

Integrating both sides of (\ref{eq:dV_simple_TOD})
 from $t_k$ to $t$, we have

\begin{IEEEeqnarray}{l}V_e(t)+\sum_{i=1}^N\int_{t_k}^t\frac{\mu_i}{2}|\xi_i(s)|^2ds
\nonumber\\\leq\sum_{i=1}^N\int_{0}^t\frac{2}{\mu_i}(|f_m(s)|^2+|f_{si}(s)|^2)ds+V_e(0)
\nonumber\\\quad
-\sum_{i=1}^N\int_{0}^{t_k}\frac{\mu_i}{2}|\xi_i(s)|^2ds,
\label{eq:dV_ditui_TOD}\end{IEEEeqnarray}
thus
\begin{align*}
V_e\!(t\!)\!\leq\!& V_e\!(0\!)
\!-\!\sum_{i=1}^N\int_{0}^t(\frac{\mu_i}{2}|\xi_i\!(s\!)|^2\!-\!\frac{2}{\mu_i}(|f_m(s)|^2\!+\!|f_{si}(s)|^2))ds,\\
&t\in[t_k,t_{k+1}).
\end{align*}
Thus, together with  (\ref{eq:Ve_positivity}),  one has $\beta (|\dot{q}_m|^2+|\dot{q}_s|^2+|e|^2) \leq V_e(t)\leq V_e(0)$ for all ${\xi _i}(t) > \frac{2}{{{\mu _i}}}\sqrt {(\Delta _m^2 + \Delta _{si}^2)}$, which yields that $\dot{q}_m,\dot{q}_{si},q_m-\check{q}_{si}\in\mathcal{L}_\infty, i=1,2,...,N$.
By now, Claim 2 is established.

 If $f_m,f_{si}\in\mathcal{L}_2$, then (\ref{eq:dV_ditui_TOD}) implies that 
 \begin{equation*}
  V_e(t)\!+\!\sum_{i=1}^N\int_{0}^t\frac{\mu_i}{2}|\xi_i(s)|^2ds\leq \sum_{i=1}^N\frac{2}{\mu_i}(\|f_m\|_2^2\!+\!\|f_{si}\|_2^2)\!+\!V_e(0),
 \end{equation*}
 thus $V_e(t)\in \mathcal{L}_\infty$, $\xi_i\in\mathcal{L}_2$, then $\dot{q}_m, \dot{q}_{si}$, $q_m-\check{q}_{si}, \eta_i\in\mathcal{L}_\infty$ and $\dot{q}_m$, $\dot{q}_{si}$,  $\int_{t-d(t)}^t\dot{q}_m(s)ds$, $\int_{t-d(t)}^t\dot{q}_{si}(s)ds, \rho_{i_k^*}\in\mathcal{L}_2$, hence with (\ref{eq:eta_TOD}) one has that $\eta_i\in\mathcal{L}_\infty\cap\mathcal{L}_2$. Thus by (\ref{eq:positionerror_transform1_TOD}-\ref{eq:positionerror_transform2_TOD}), we know that $q_m-\check{q}_{si}(t-d_{si}(t)),\check{q}_{si}-q_m(t-d_m(t))\in\mathcal{L}_\infty$. By the closed-loop dynamics (\ref{eq:clo_system_TOD}), we have $\ddot{q}_m,\ddot{q}_{si}\in\mathcal{L}_\infty$.
The remainder of the proof for Claim 3 can be obtained by following the same line of reasoning for Claim 1.

The proof of the last claim is obvious and is thus omitted here.

  \end{proof}

  \section{Simulations}\label{sec:simulation}
  In this section, numerical examples and simulation results are given to verify the effectiveness of the proposed results.
  \subsection{Numerical studies}
  Consider a  teleoperation system (\ref{eq:master}-\ref{eq:slave})
 with  $N=2$ and $N=3$, respectively.
  To verify the results of Theorem~\ref{Thm:RR}  and Theorem~\ref{Thm:TOD},  the max.  allowable $\!MATI\!$ is calculated by fixing
 the upper bound of communication delays $MAD$ as $0$s, $0.2$s and $0.5$s, respectively.
 For simplicity, we first analyze the stability of the closed-loop system with $K_i^p=K_i^d=20I, i=m,s1,s2,...,sN$. The max. allowable values of $MATI$ that preserve the stability are given in Table~\ref{tab:MATI_sim_N2N3} and the resulting max. $h_S$ in Theorem~\ref{Thm:RR} and $h_M$ in Theorem~\ref{Thm:TOD} are shown in Table~\ref{tab:h_sim_N2N3}. From Table~\ref{tab:MATI_sim_N2N3}, we find that the resulting max. $MATI$  under RR protocol is larger than the one under TOD protocol. Besides, the max. allowable $MATI$  under each scheduling protocol decreases with the increase of $N$.  We  further fix the number of slaves as $N=3$, and analyze the stability of the closed-loop system with $K_1^p=10I$, $K_2^p=20I$, $K_3^p=30I$, and $K_i^d=20I$ for $i=m, s1, s2, s3$.  The max. allowable values of $MATI$ that preserve the stability are given in Table~\ref{tab:MATI_N3} and the resulting max. $h_S$ and $h_M$ are shown in Table~\ref{tab:h_N3}. From both of Table~\ref{tab:MATI_sim_N2N3} and Table~\ref{tab:MATI_N3}, it is found that the max. allowable $MATI$ decreases under the scheduling protocols with the increase of the upper bounds of the communication delays.


\begin{table}[h!]\centering
\caption{max. allowable $MATI$ for $K_i^p=20I$, $K_i^d=20I$, $i=m, s1, .., sN$}
\begin{tabular}{cccc@{}cccc}
\hline
\multirow{2}{*}{$MAD$}& \multicolumn{3}{c}{$N=2$}&&\multicolumn{3}{c}{$N=3$}\\
\cline{2-4}\cline{6-8}
& 0& 0.2&0.5 && 0& 0.2&0.5\\
\hline
Theorem \ref{Thm:RR} &0.6666& 0.5333 &0.3333&&0.5 & 0.4 &0.25  \\
Theorem \ref{Thm:TOD} &0.4531 & 0.2431 &-&&0.2411 & 0.0411 &- \\
\hline
\end{tabular}\label{tab:MATI_sim_N2N3}
\end{table}

\begin{table}[h!]\centering
\caption{max. time delays  for $K_i^p=20I$, $K_i^d=20I$, $i=m, s1,..., sN$}
\begin{tabular}{cccc@{}cccc}
\hline
\multirow{2}{*}{$h_S/h_M$}& \multicolumn{3}{c}{$N=2$}&&\multicolumn{3}{c}{$N=3$}\\
\cline{2-4}\cline{6-8}
& 0& 0.2&0.5 && 0& 0.2&0.5\\
\hline
$h_S$ in &\multirow{2}{*}{1.3332}&\multirow{2}{*}{1.2666}&\multirow{2}{*}{1.6666}&&\multirow{2}{*}{1.5} & \multirow{2}{*}{1.4}&\multirow{2}{*}{1.25}\\Theorem~\ref{Thm:RR}  \\
$h_M$ in&\multirow{2}{*}{0.4531}&\multirow{2}{*}{0.4531}&\multirow{2}{*}{-}&&\multirow{2}{*}{0.2411} & \multirow{2}{*}{0.2411} &\multirow{2}{*}{-} \\ Theorem~\ref{Thm:TOD}\\
\hline
\end{tabular}\label{tab:h_sim_N2N3}
\end{table}

\begin{table}[h!]\centering
\caption{max. allowable $MATI$ for $K_1^p=10I$, $K_2^p=20I$, $K_3^p=30I$, $K_i^d=20I$, $i=m, s1, s2, s3$}\begin{tabular}{cccc}
\hline
$MAD$ & 0& 0.2&0.5  \\
\hline
Theorem \ref{Thm:RR} &0.3333 & 0.2333 &0.1 \\
Theorem \ref{Thm:TOD}&0.2066 & 0.0066 &-\\
\hline
\end{tabular}\label{tab:MATI_N3}
\end{table}

\begin{table}[h!]\centering
\caption{max. allowable time delays for $K_1^p=10I$, $K_3^p=20I$, $K_3^p=30I$, $K_i^d=20I$, $i=m, s1, s2, s3$}
\begin{tabular}{cccc}
\hline
$h_S/h_M$ & 0& 0.2&0.5  \\
\hline
$h_S$ in Theorem \ref{Thm:RR} &1 & 0.8999&0.8 \\

$h_M$ in Theorem \ref{Thm:TOD} &0.2066 & 0.2066 &-\\
\hline
\end{tabular}\label{tab:h_N3}
\end{table}

\begin{figure}[t!]
  \centering
  \includegraphics[width=0.8\linewidth]{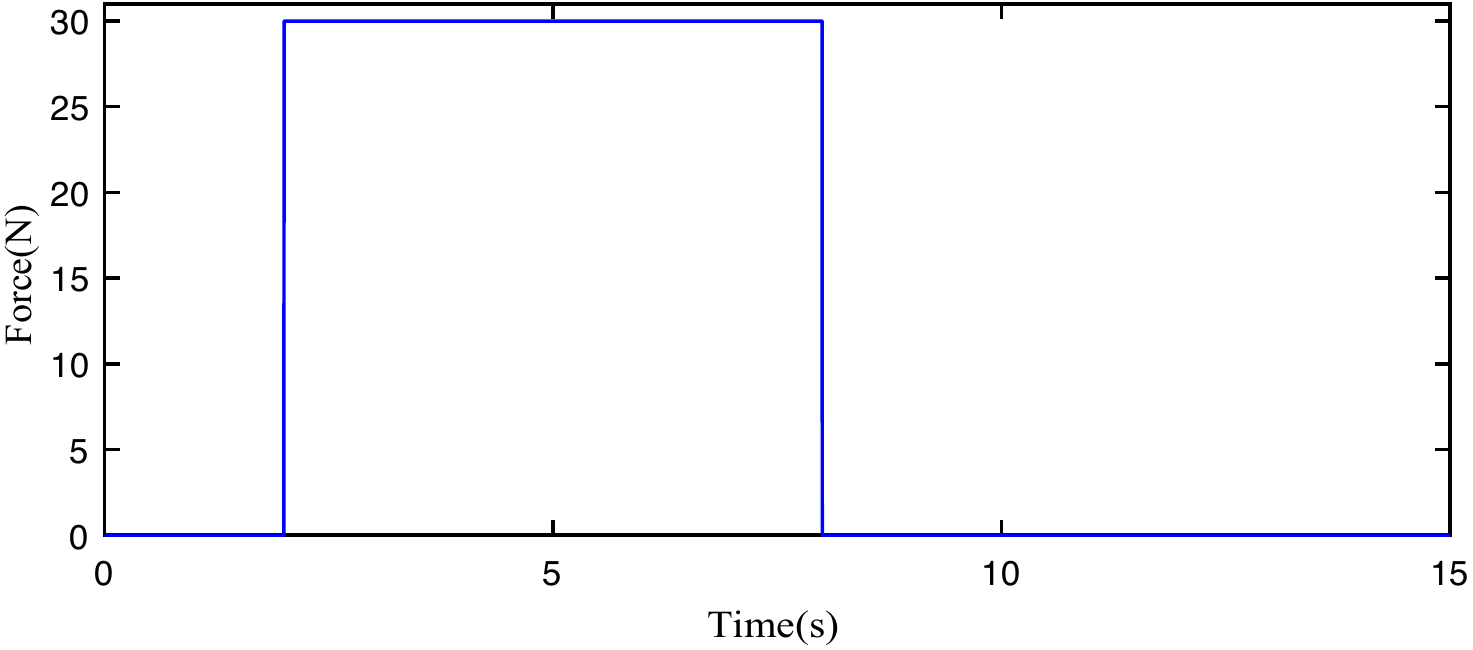}\\
  \caption{External force $F_2(t)$}.\label{fig:f}
\end{figure}

\subsection{Simulation results}
\begin{figure}[t!]\centering
\begin{minipage}{0.48\linewidth}
\centerline{ \includegraphics[width=\linewidth]{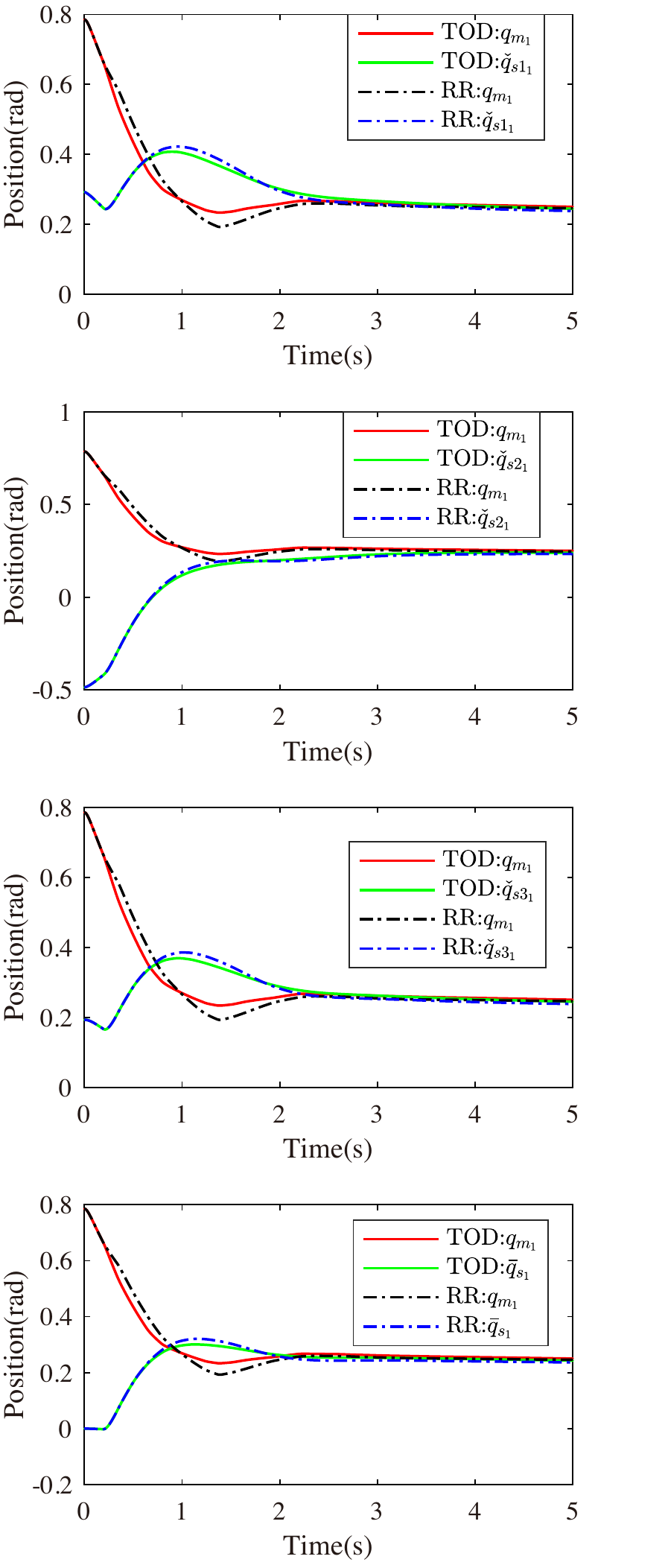}}
  \centerline{(a)}
\end{minipage}
\begin{minipage}{0.48\linewidth}
\includegraphics[width=\linewidth]{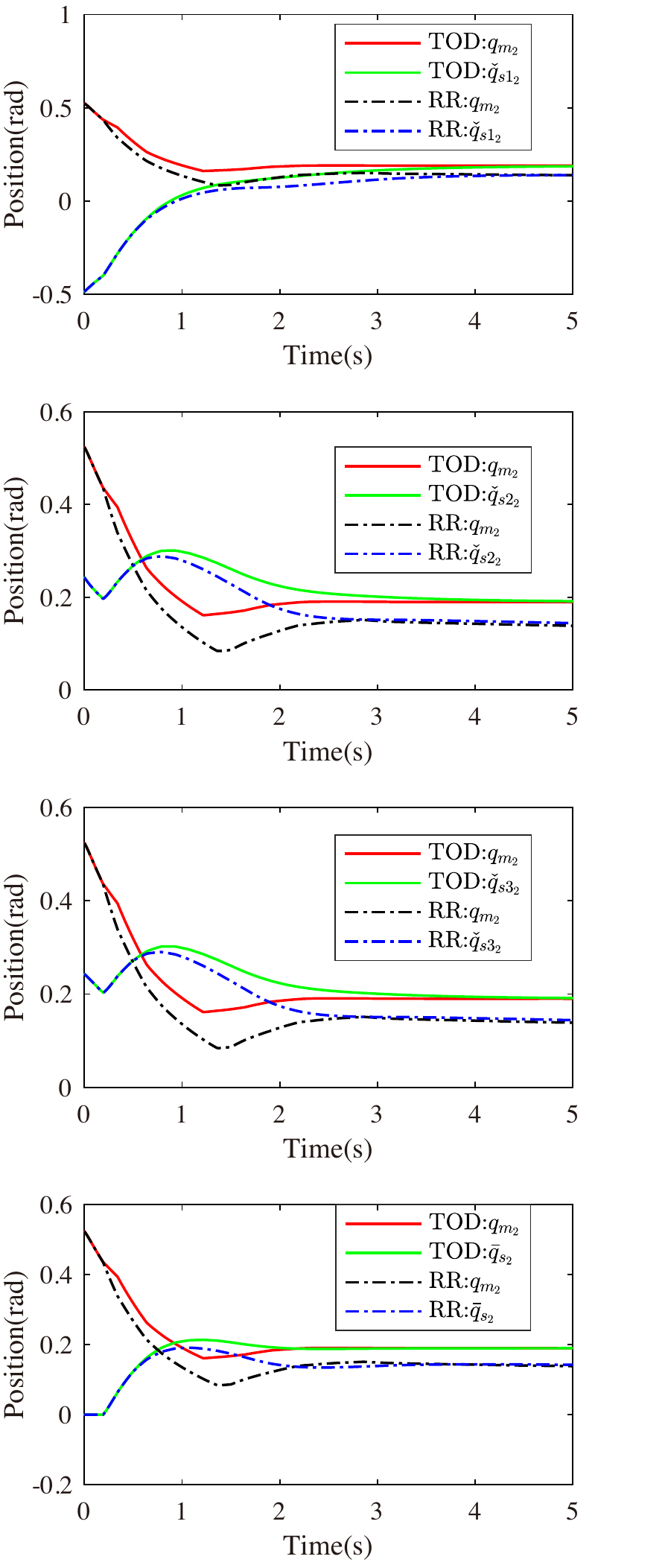}\\
   \centerline{(b)}
\end{minipage}
\caption{Scenario 1: joint positions of the master and the slave manipulators. (a) joint 1, (b) joint 2. }
\label{fig:pos_free}
\end{figure}

\begin{figure}[t!]\centering
\begin{minipage}{0.48\linewidth}
\centerline{ \includegraphics[width=\linewidth]{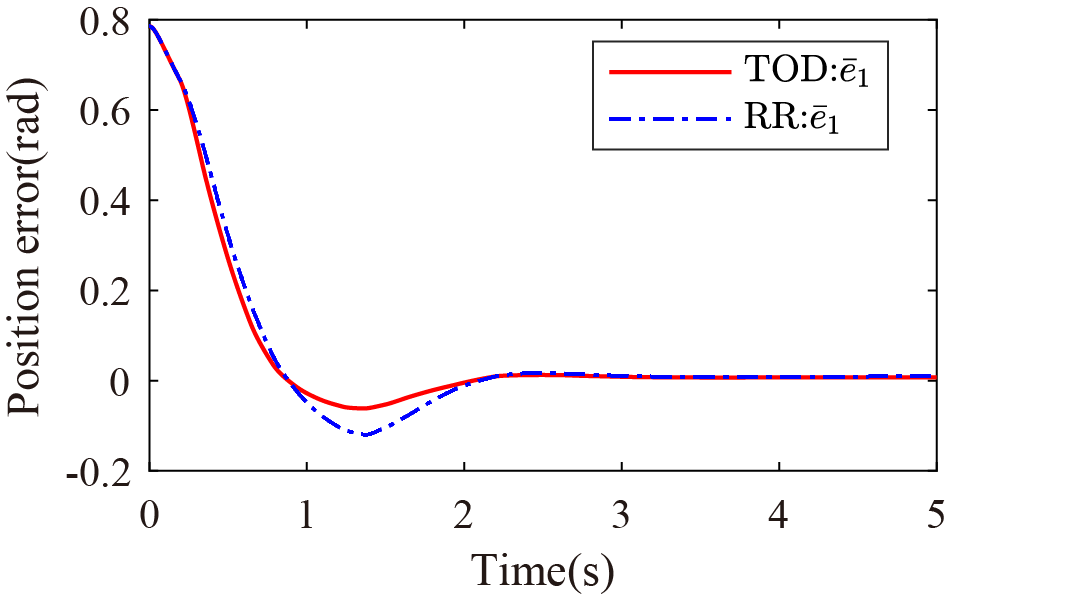}}
  \centerline{(a)}
\end{minipage}
\begin{minipage}{0.48\linewidth}
\includegraphics[width=\linewidth]{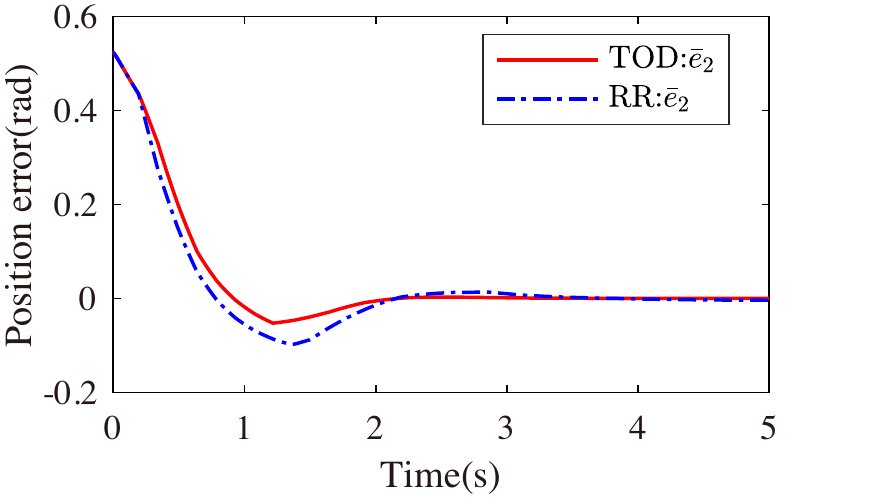}\\
   \centerline{(b)}
\end{minipage}
\caption{Scenario 1: joint position errors between the master and the slave manipulators. (a) joint 1, (b) joint 2. }
\label{fig:pos_error_free}
\end{figure}

Simulation studies on a  SMMS teleoperation system  with one 2-DOF master manipulator and three 2-DOF slave manipulators have been performed. The parameters of the dynamic models in (\ref{eq:master}-\ref{eq:slave})  are given as:
\begin{align*}
&M_i(q_i)=\begin{bmatrix}
M_{i_{11}}(q_i)&M_{i_{12}}(q_i)\\
M_{i_{21}}(q_i)&M_{i_{22}}(q_i)
\end{bmatrix}, G_i(q_i)=\begin{bmatrix}
G_{i_1}(q_i) \\
G_{i_2}(q_i)
\end{bmatrix},\\&C_i(q_i,\dot{q}_i)=\begin{bmatrix}
C_{i_{11}}(q_i,\dot{q}_i)&C_{i_{12}}(q_i,\dot{q}_i)\\
C_{i_{21}}(q_i,\dot{q}_i)&C_{i_{22}}(q_i,\dot{q}_i)
\end{bmatrix},
\end{align*}
where
$M_{i_{11}}(q_i)\!=\!l_{i_2}^2m_{i_2}\!+\!l_{i_1}^2(m_{i_1}+m_{i_2})\!+\!2l_{i_1}l_{i_2}m_{i_2}\cos(q_{i_2})$,
$M_{i_{22}}(q_i)\!=\!l_{i_2}^2m_{i_2}$,
$M_{i_{12}}(q_i)\!=\!M_{i_{21}}(q_i)=l_{i_2}^2m_{i_2}\!+\!l_{i_1}l_{i_2}m_{i_2}\cos(q_{i_2})$,
$C_{i_{11}}(q_i,\dot{q}_i)\!=\!-l_{i_1}l_{i_2}m_{i_2}\sin(q_{2_i})\dot{q}_{i_2}$,
$C_{i_{12}}(q_i,\dot{q}_i)\!=\!-l_{i_1}l_{i_2}m_{i_2}\sin(q_{2_i})(\dot{q}_{1_i}\!+\!\dot{q}_{2_i})$,
$C_{i_{21}}(q_i,\dot{q}_i)\!=\!l_{i_1}l_{i_2}m_{i_2}\sin(q_{2_i})\dot{q}_{1_i},C_{i_{22}}(q_i,\dot{q}_i)=0$,
$G_{i_1}(q_i)\!=\!\frac{1}{l_{i_2}}gl_{i_2}^2m_{i_2}\cos(q_{1_i}+q_{2_i})\!+\!
\frac{1}{l_{i_1}}(l_{i_2}^2m_{i_2}+l_{i_1}^2(m_{i_1}+m_{i_2})-l_{i_2}^2m_{i_2})\cos(q_{1_i})$,
$G_{i_2}(q_i)\!=\!\frac{1}{l_{i_2}}gl_{i_2}^2m_{i_2}\cos(q_{1_i}+q_{2_i})$,
with $i=\{m, s1, s2, s3\}$ representing the master manipulator, the first, the second and the third slave manipulators, respectively.
The mass of the manipulators are chosen as $m_{i_1}=1$kg, $m_{i_2}=0.5$kg, the length of links for the master and the slave robots are $l_{i_1}=0.5$m, $l_{i_2}=0.3$m. The controller parameters are set as $K_i^p=20I, K_i^d=20I$ for $i=m, s1, s2, s3$ in the following simulations. The communication delays are set as $T_k=0.04+0.06|\sin(s_k)|$.
The sampling interval is chosen as $1.14$s, which is less than the max. allowable $MATI$ under both RR  and TOD scheduling protocols according to Table~\ref{tab:MATI_sim_N2N3}. We further assume that the relative distance from the formation's center to each slave manipulator is $\gamma_1=[0.1, -0.3]^T, \gamma_2=[-0.3, 0.15]^T, \gamma_3=[0.2, 0.15]^T$, respectively.  For the considered SMMS teleoperation system  under RR and TOD scheduling protocols, respectively, the following simulation set-ups are considered:
\begin{enumerate}
\item Scenario 1: The teleoperation system is driven by no external forces. That is, the external forces $f_m(t)=f_{si}(t)\equiv 0$. Under this circumstance, the initial conditions for the master and the three slave manipulators are chosen as  $q_m(t)=[\frac{\pi}{4},\frac{\pi}{6}]^T, q_{s1}(t)=[\frac{\pi}{8},-\frac{\pi}{4}]^T$, $q_{s2}(t)=[-\frac{1}{4}\pi,\frac{1}{8}\pi]^T$, $q_{s3}=[\frac{1}{8}\pi,\frac{1}{8}\pi]^T$, and $\dot{q}_m(t)=\dot{q}_{s1}(t)=\dot{q}_{s2}(t)=\dot{q}_{s3}(t)=[0, 0]^T$, $\ddot{q}_m(t)=\ddot{q}_{s1}(t)=\ddot{q}_{s2}(t)=\ddot{q}_{s3}(t)=[0, 0]^T$.

\begin{figure}[t!]\centering
\subfigure[]{
\label{Fig.sub.1}
\includegraphics[width=0.48\linewidth]{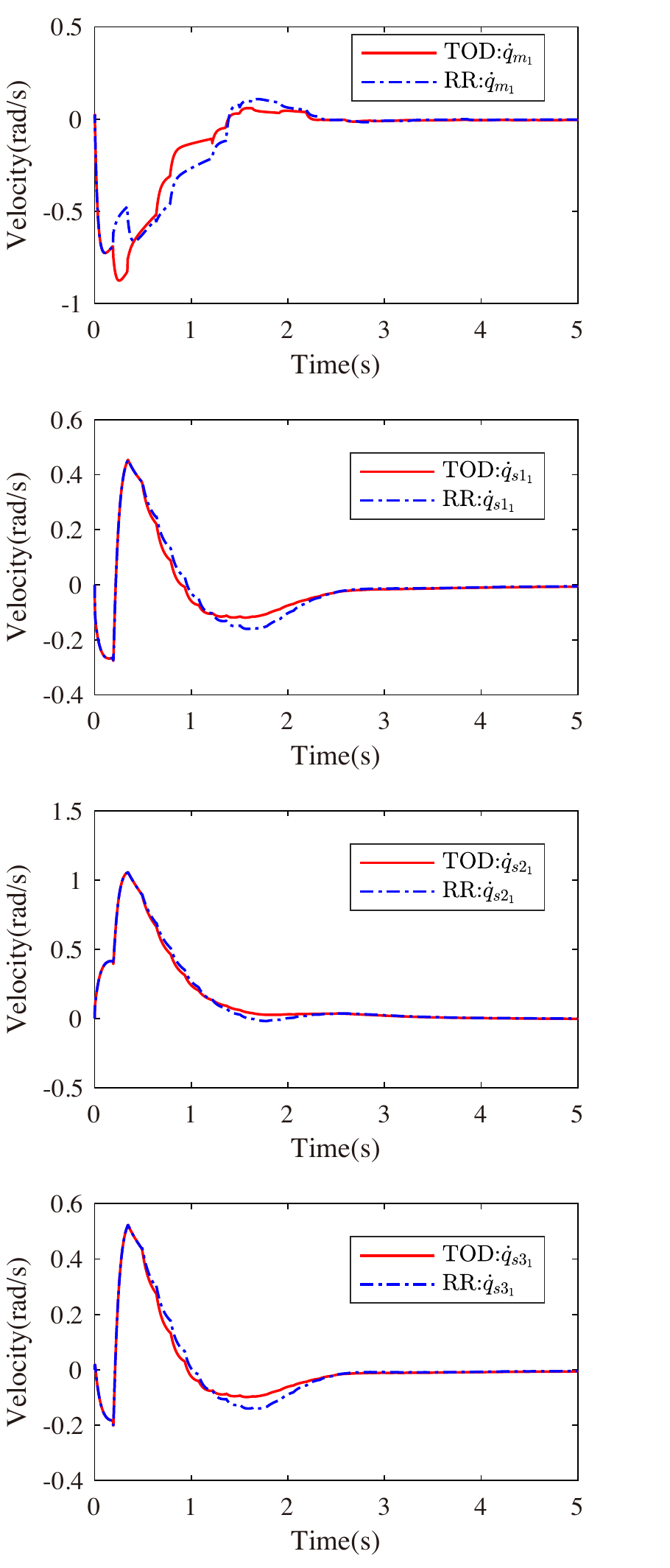}}
\subfigure[]{
\label{Fig.sub.2}
\includegraphics[width=0.48\linewidth]{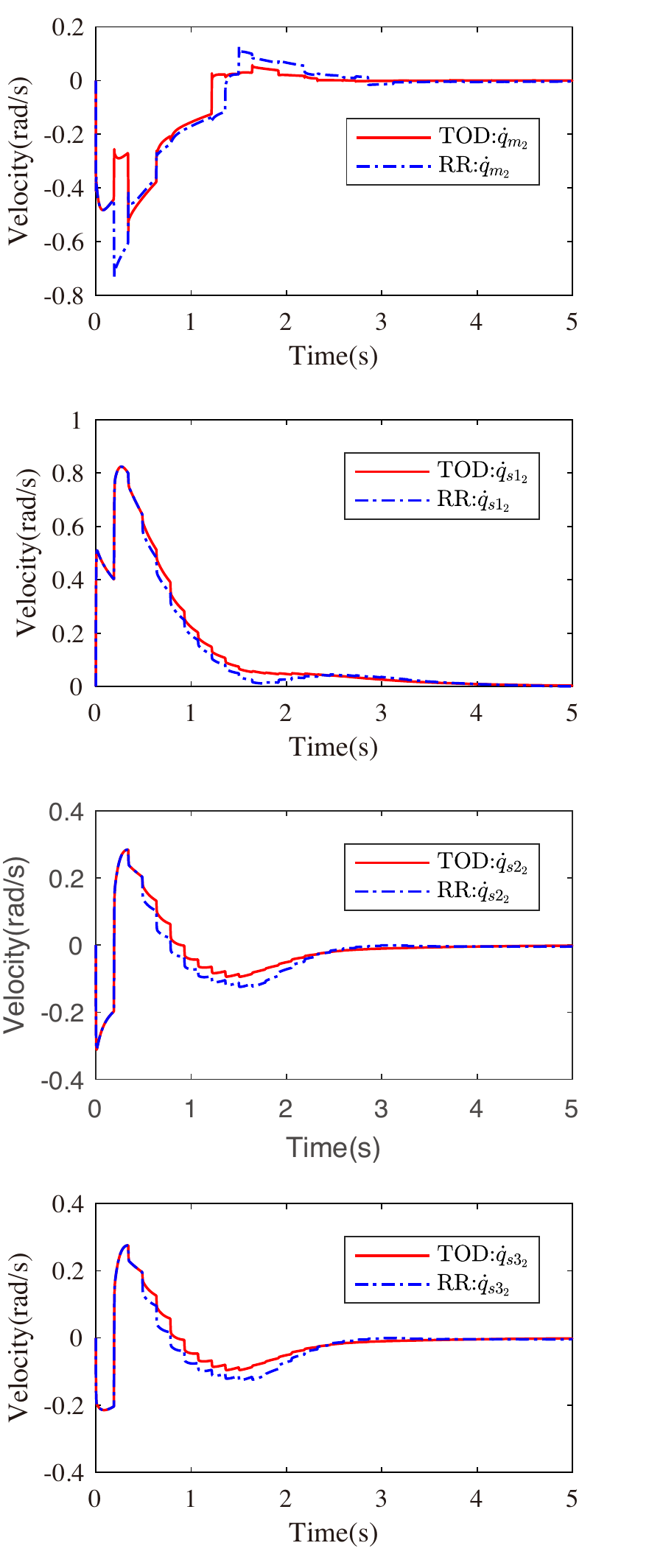}}
\caption{Scenario 1: joint velocities  of the master and the slave manipulators.(a) joint 1, (b) joint 2. }
\label{fig:vel_free}
\end{figure}

\item Scenario 2: The master's end-effector is driven by a bounded force  while the slaves move in free motion. We assume that a force $F_1(t)=25+10\sin(t)$ is exerted to the master manipulator at the end-effector. Thus the torque applied by the human operator is $f_m=J_m^T[0,1]^T F_1(t)$, where
$J_m=\begin{bmatrix}  J_{m_{11}}&J_{m_{12}}\\J_{m_{21}}&J_{m_{22}}\end{bmatrix}$with $J_{m_{11}}=-L_{m_1}\sin(q_{m_1})-L_{m_2}\sin(q_{m_1}+q_{m_2})$ , $J_{m_{12}}=-L_{m_2}\sin(q_{m_1}+q_{m_2})$, $J_{m_{21}}=L_{m_1}\cos(q_{m_1})+L_{m_2}\cos(q_{m_1}+q_{m_2})$, $J_{m_{22}}=L_{m_2}\cos(q_{m_1}+q_{m_2})$. Clearly, $F_1(t)\in\mathcal{L}_\infty$ and thus $f_m(t)\in\mathcal{L}_\infty$.
 For simplicity, the initial conditions for the master and the slaves are set as zeros, i.e.,  $q_i(t)=\dot{q}_{i}(t)=\ddot{q}_i=[0, 0]^T$, $i=m,s1,s2,s3$, repectively.

 \item Senario 3: The master's end-effector is driven by a  rectangle signal $F_2(t)$ depicted in Fig.~\ref{fig:f}, Obviously, it belongs to $\mathcal{L}_2$. In the slave side, there is a stiff wall at $y=0.3m$, and the contact torque is expressed as $f_{si}(t)=-J_{si}^T[0,1]^T 10000(y-0.3)$Nm.
 The initial condition for each manipulator is still assumed to be zero in this case.
\end{enumerate}

The simulation results for the considered teleoperation system (\ref{eq:master}-\ref{eq:slave}) under RR scheduling protocol (\ref{eq:ikstar_RR}-\ref{eq:ik_RR})  and TOD scheduling protocol (\ref{eq:Def_TOD}), respectively, are  depicted in Fig.~\ref{fig:pos_free} -- Fig.~\ref{fig:tau_contact}, and it is observed that the closed-loop system is stable under the scheduling protocols in all the simulation circumstances.

\begin{figure}[h!]\centering
\begin{minipage}{0.48\linewidth}
\centerline{ \includegraphics[width=\linewidth]{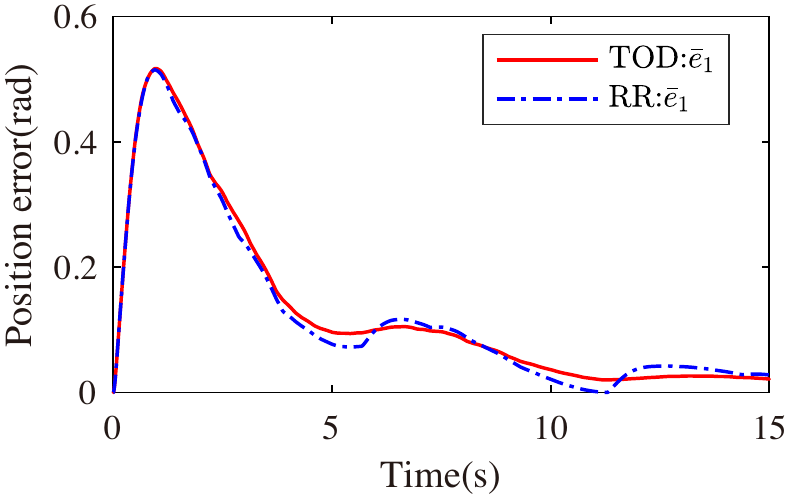}}
  \centerline{(a)}
\end{minipage}
\begin{minipage}{0.48\linewidth}
\includegraphics[width=\linewidth]{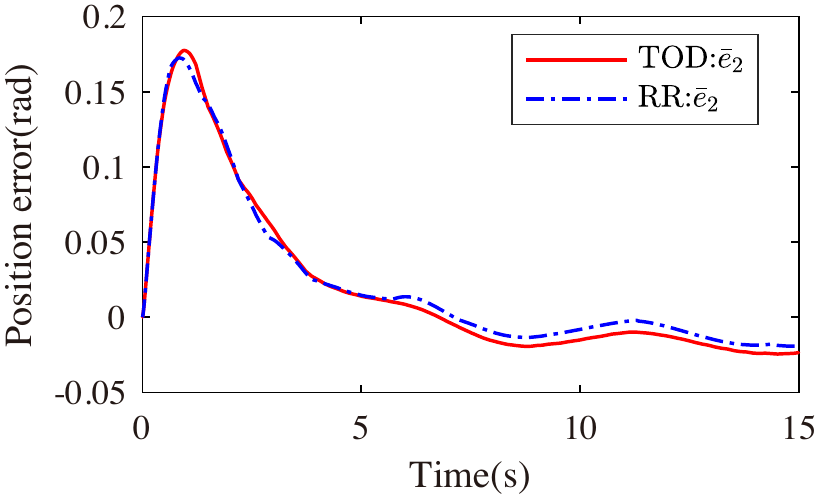}\\
   \centerline{(b)}
\end{minipage}
\caption{Scenario 2: joint position errors between the master and the slave manipulators. (a) joint 1, (b) joint 2. }
\label{fig:pos_error_Linf}
\end{figure}


\begin{figure}[h!]\centering
\begin{minipage}{0.48\linewidth}
\centerline{ \includegraphics[width=\linewidth]{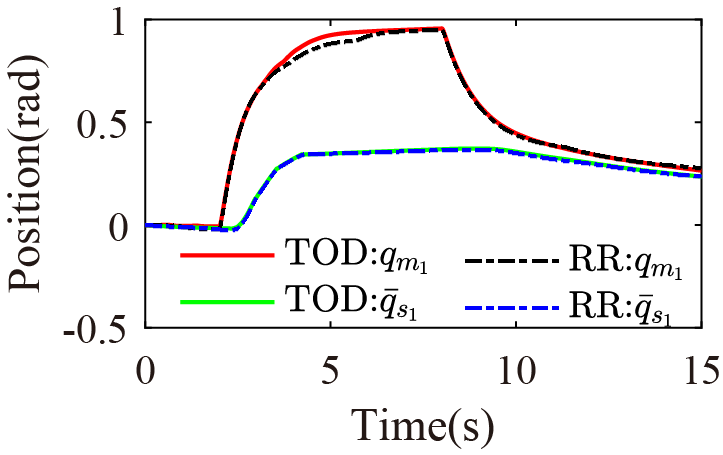}}
  \centerline{(a)}
\end{minipage}
\begin{minipage}{0.48\linewidth}
\includegraphics[width=\linewidth]{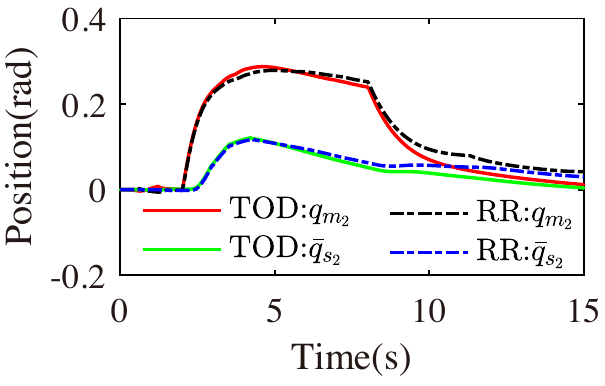}\\
   \centerline{(b)}
\end{minipage}
\caption{Scenario 3: joint positions of the master and the slave manipulators. (a) joint 1, (b) joint 2. }
\label{fig:pos_contact}
\end{figure}

\begin{figure}[h!]\centering
\begin{minipage}{0.48\linewidth}
\centerline{ \includegraphics[width=\linewidth]{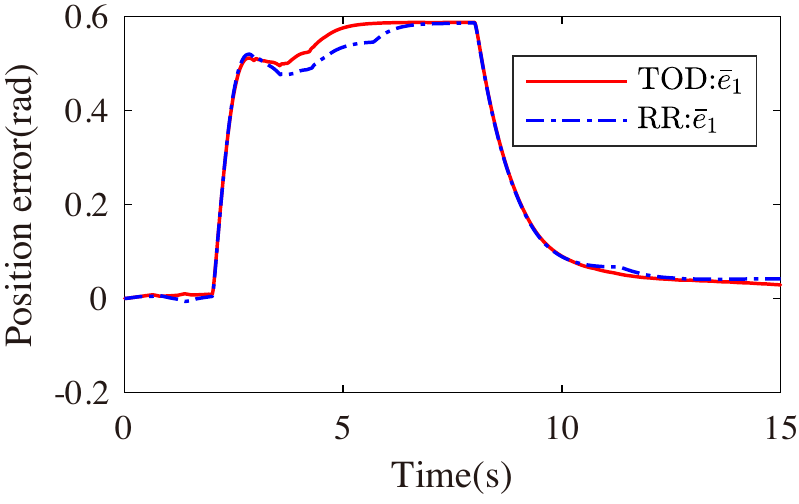}}
  \centerline{(a)}
\end{minipage}
\begin{minipage}{0.48\linewidth}
\includegraphics[width=\linewidth]{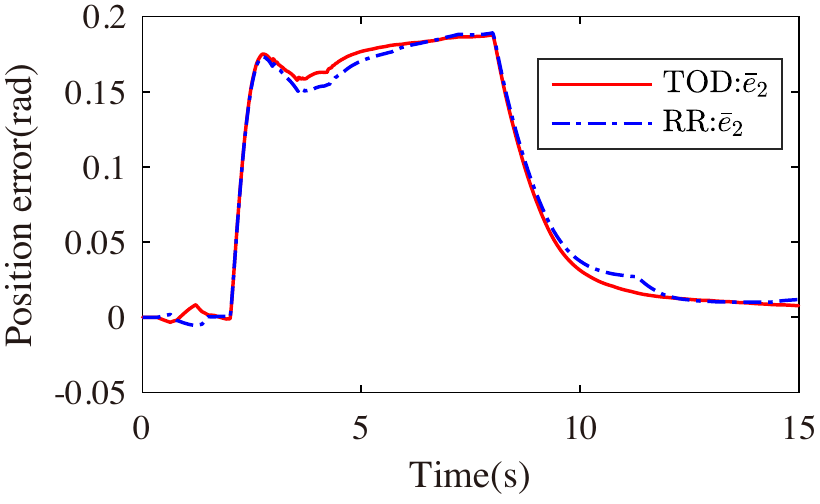}\\
   \centerline{(b)}
\end{minipage}
\caption{Scenario 3: joint position errors between the master and the slave manipulators. (a) joint 1, (b) joint 2. }
\label{fig:pos_error_contact}
\end{figure}

\begin{figure}[h!]\centering
\begin{minipage}{0.48\linewidth}
\centerline{ \includegraphics[width=1\linewidth]{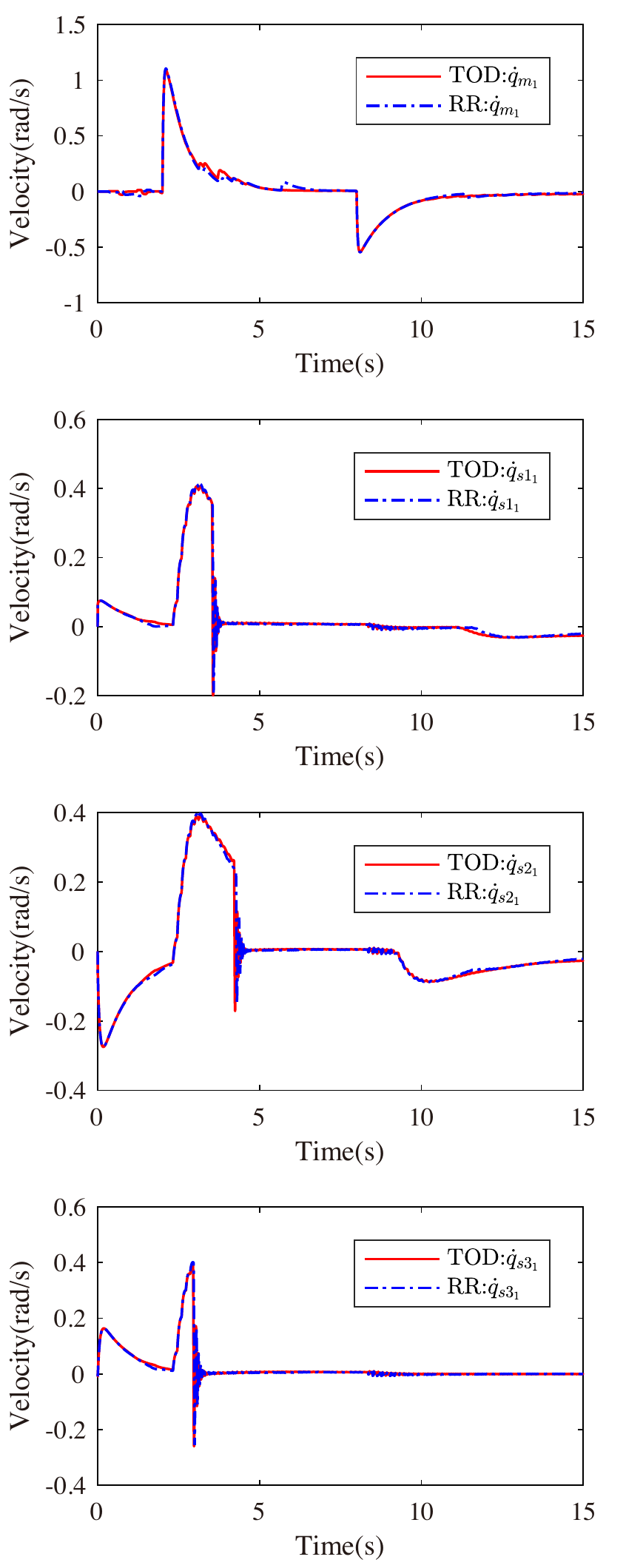}}
  \centerline{(a)}
\end{minipage}
\begin{minipage}{0.48\linewidth}
\includegraphics[width=1\linewidth]{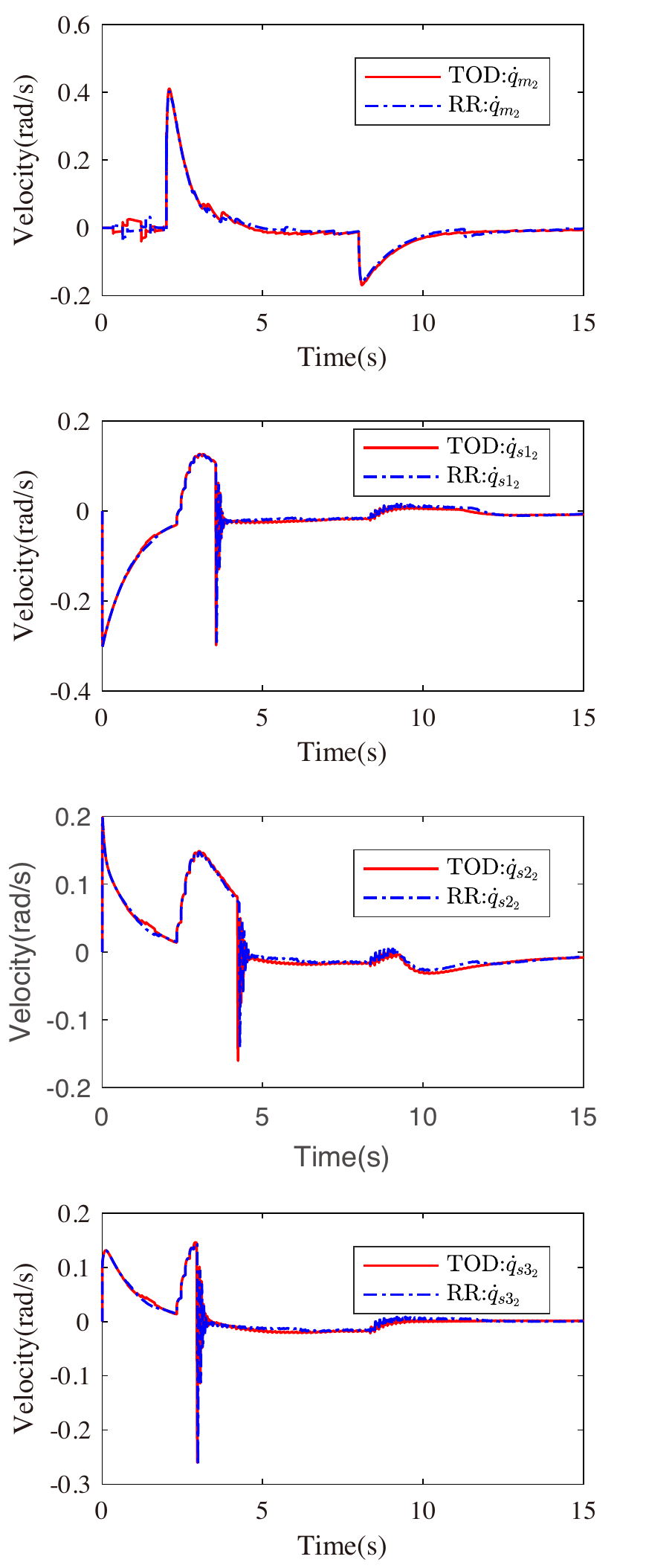}\\
   \centerline{(b)}
\end{minipage}
\caption{Scenario 3: joint velocities  of the master and the slave manipulators. (a) joint 1, (b) joint 2. }
\label{fig:vel_contact}
\end{figure}

We now analyze the simulation results of the teleoperation system in different simulation scenarios. Firstly, when the considered teleoperation system is in free motion, i.e., simulation scenario 1, the simulation results are given in Fig.~\ref{fig:pos_free} -- Fig.~\ref{fig:vel_free}. From Fig. \ref{fig:pos_free}, we know that the formation's center of the slaves follows the master's position at around 2s (the last row of Fig.~\ref{fig:pos_free}) under the scheduling protocols. The tracking performance of each pair of $q_m$ and $\check{q}_{si}$ is also provided in Fig.~\ref{fig:pos_free}. It is noted that in Fig.~\ref{fig:pos_free}, the positions of the manipulators under TOD scheduling protocol are higher than the ones under RR scheduling protocol. Fig. \ref{fig:pos_error_free} depicts the curves of position errors, and it is shown that the position errors under TOD scheduling protocol converge faster to the origin in this scenario. 
From Fig.~\ref{fig:vel_free}, we observe that the velocities of each manipulator converge to the origin very fast (after about 3s).

\begin{figure}[h!]\centering
\begin{minipage}{0.48\linewidth}
\centerline{ \includegraphics[width=1\linewidth]{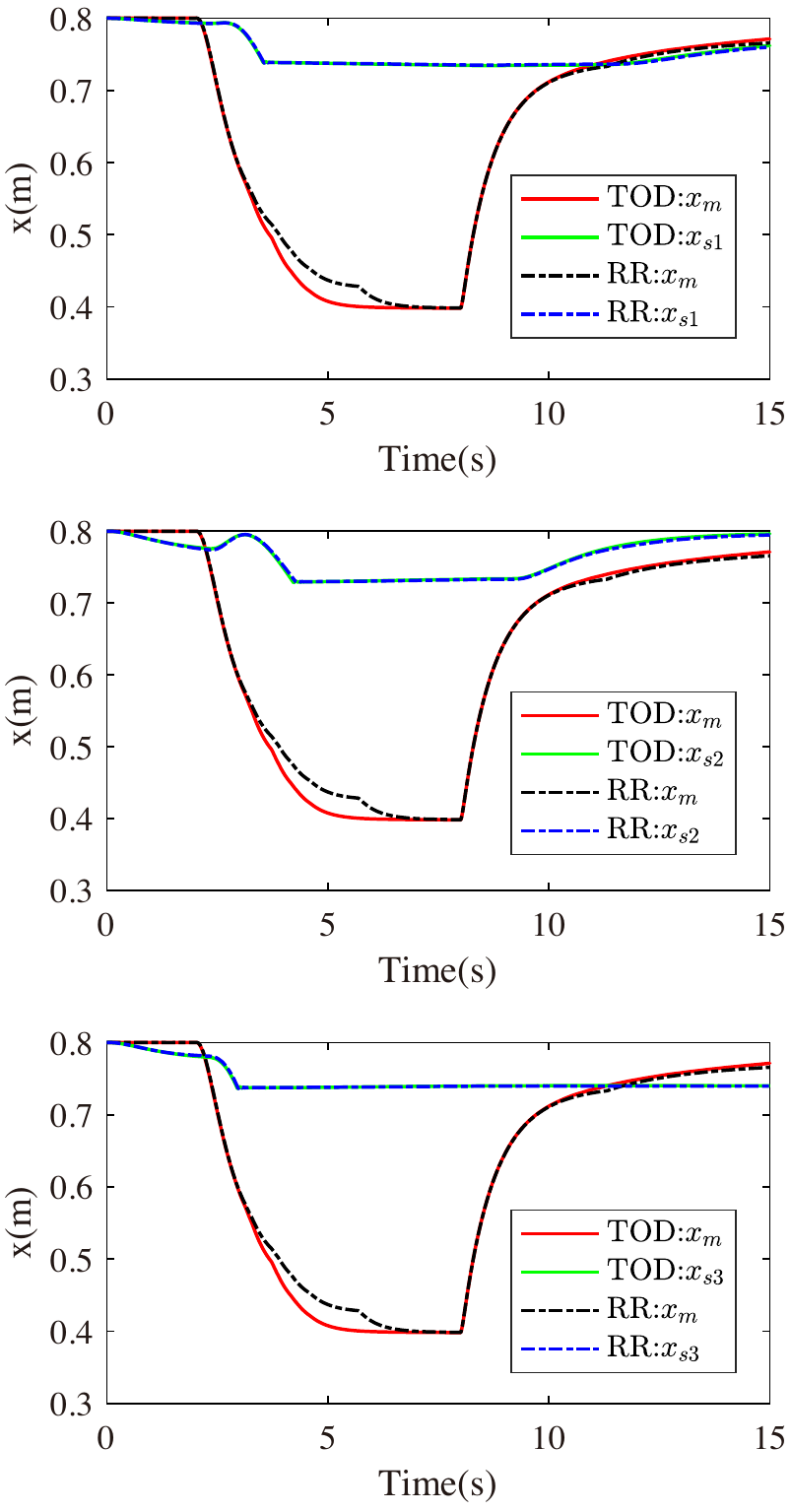}}
  \centerline{(a)}
\end{minipage}
\begin{minipage}{0.48\linewidth}
\includegraphics[width=1\linewidth]{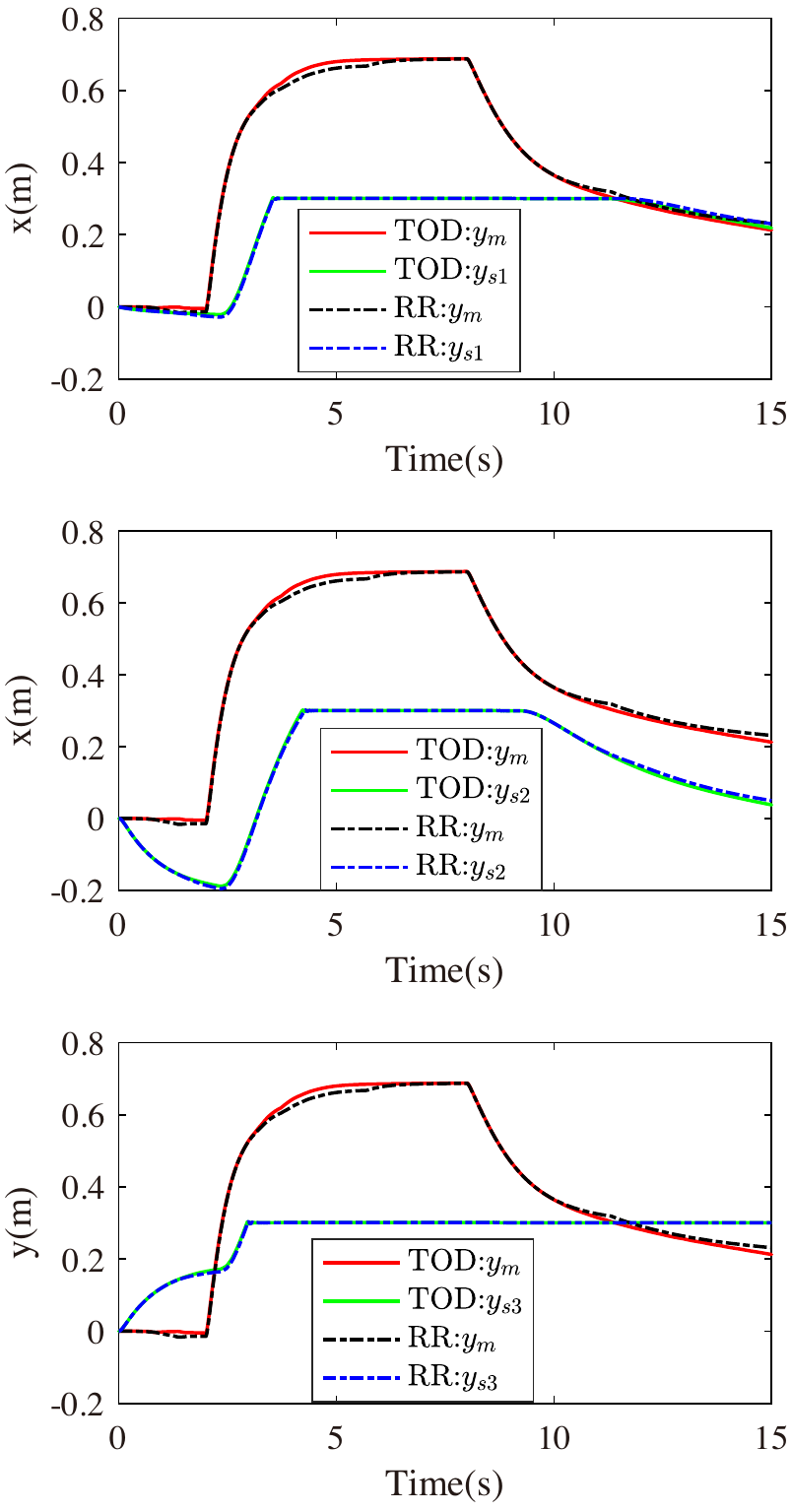}\\
   \centerline{(b)}
\end{minipage}
\caption{Scenario 3: end-effector positions of master and slave manipulators. (a) x-axis, (b) y-axis. }
\label{fig:endpos_contact_RR}
\end{figure}


\begin{figure}[t!]\centering
\subfigure[]{
\label{Fig.sub.1}
\includegraphics[width=0.48\linewidth]{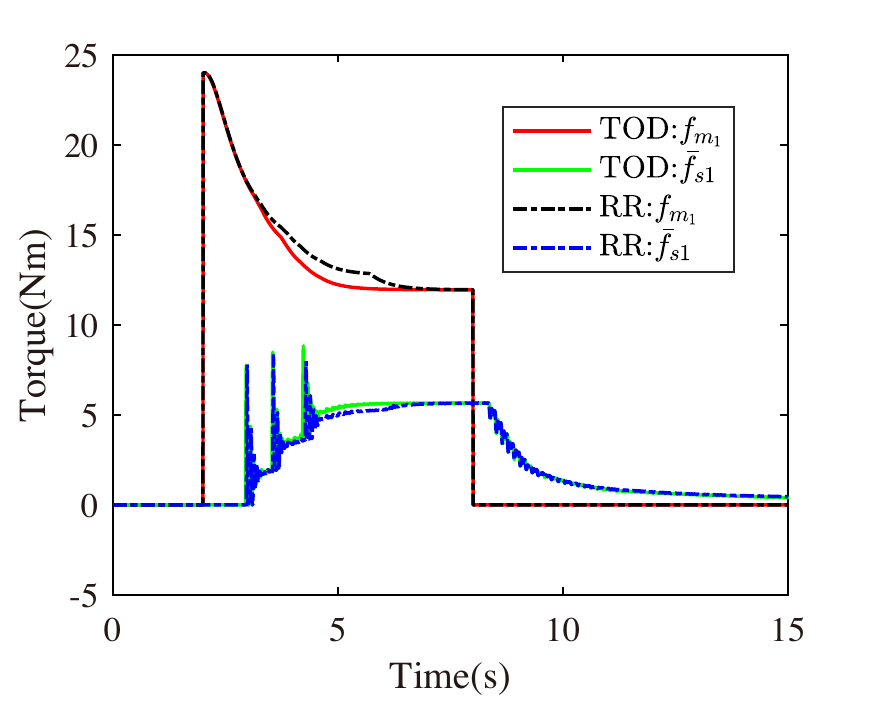}}
\subfigure[]{
\label{Fig.sub.2}
\includegraphics[width=0.48\linewidth]{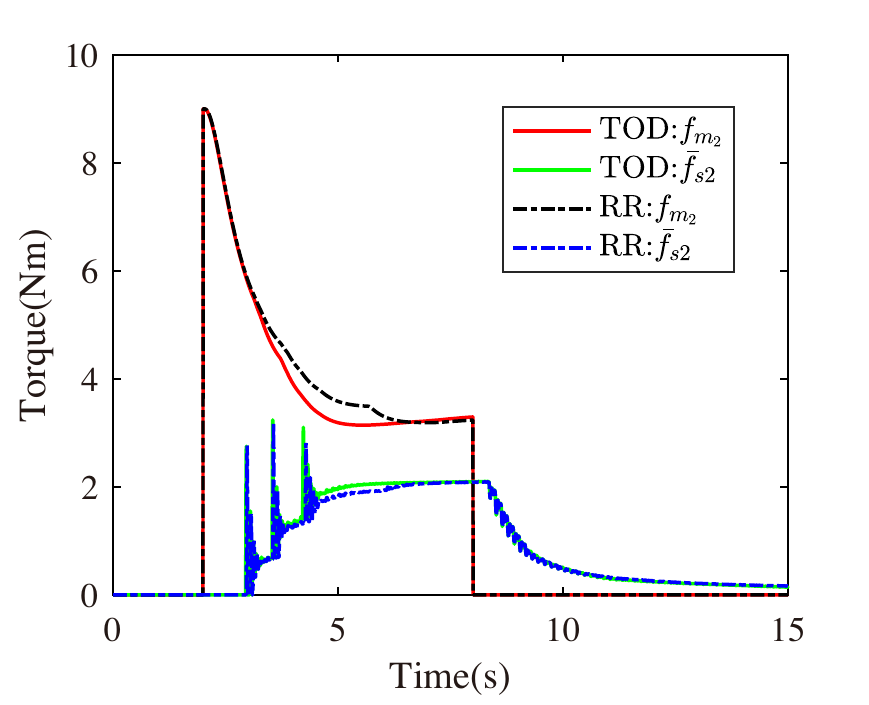}}
\caption{Scenario 3: external forces for master and slave manipulators. (a) joint 1, (b) joint 2. }
\label{fig:tau_contact}
\end{figure}

Secondly, the considered teleoperation system moves in Scenario 2.
Due to the space limitation, we only provide Fig. \ref{fig:pos_error_Linf} which shows the position errors of both links in this scenario. It is observed that the position errors between the master and the slaves are different when different scheduling protocols are employed, however, compared with the results in Fig.~\ref{fig:pos_free}~-~Fig.~\ref{fig:pos_error_free}, it is not easy to conclude which one is better than the other in this simulation scenario.  This is because that, compared with the impact of the scheduling protocols on the closed-loop system, the external forces have a far greater impact on the teleoperation system. This is also applicable for the  simulation in Scenario~3, for which both the human force and environmental forces are considered.

Finally, the simulation results can be found in Fig.~\ref{fig:pos_contact}~--~Fig.~\ref{fig:vel_contact} when the master is driven by a rectangle signal and all the slaves contact or at least one slave is in contact with a rigid wall. It is found that from  Fig.~\ref{fig:pos_contact}, the position of the master and the formation center of the slaves converge to each other after around the time $t=12s$. From Fig.~\ref{fig:vel_contact}, it can be seen that the joint velocities converge to the origin around the time $t=12s$ as well. During 4.5s to 9s, the manipulators stay stationary as the external interactions with the human operator and with the environments keep the same and the velocities are zero during this time interval.
Note that when the human force disappears, i.e., $t\geq 8$s, the joint positions under RR scheduling protocol are higher than the ones under TOD scheduling protocol, which are different from the system responses in Scenario 1. This is because the contact forces for the slaves still exist after the human force disappears (see Fig. ~\ref{fig:endpos_contact_RR}(b)), especially, the third slave is always in contact with the wall during this simulation. Hence the teleoperation system is not in free motion when $t\geq 8$s. This also confirms that the external forces have a greater impact on the teleoperation system compared with the influence of scheduling protocols.
The position tracking performance of the end-effectors and the force tracking performance under the scheduling protocols, which are depicted in Fig.~\ref{fig:endpos_contact_RR} and Fig.~\ref{fig:tau_contact}, respectively, are also provided.  From Fig.~\ref{fig:endpos_contact_RR}(b), we find that the motion of the slaves' end-effectors follow the master end-effectors' with a bias (regarding to $\gamma_i$) at the beginning (around $t<2.5$s) until one of the slaves reaches contact with the wall at $y_{si}=0.3$, then the slave(s) can not move further. The master also stays stationary after a short delay. This implies that the slaves' motions are fed back to the master through the scheduling communication network. When the slaves reach contact with the stiff wall, there is an almost static position error between the master and the slaves' formation center, this error implies the magnitude of input forces $f_{s1}, f_{s2}, f_{s3}$. Fig.~\ref{fig:tau_contact} depicts the curves of input forces, which show that the magnitude of human force $f_m$ is almost the same as the  mean value of all the environmental forces $f_{s1}, f_{s2}, f_{s3}$.


\section{Conclusion}\label{sec:conclusion}
In this paper, bilateral teleoperation of multiple slaves under scheduling communication has been investigated. RR scheduling and TOD scheduling protocols have been respectively utilized to transmit the information of multiple slaves, and only the newly-updated information of one slave can be transmitted through the communication network to the master side. The time-varying transmission time delays have been considered. With the proposed hybrid P+d controller under scheduling communication, the stability criteria in terms of LMIs,  which give the sufficient conditions related to the controller gains, the upper bound of time delays, and the maximum allowable sampling interval, have been provided by properly constructing new kinds of Lyapunov-Krasovskii functionals. The upper bound of the sampling interval and time delays can be derived to ensure the master-slave synchronization for teleoperation. Finally, numerical studies have been given and an example of teleoperation system with one master and three slaves has been provided for simulation illustration.  Future works will include the extension to MMMS teleoperation systems interacting with complex environments or completing complex tasks such as grabbing an object,  under hybrid scheduling protocols, packet dropouts and asymmetric time delays.

\section*{Acknowledgments}
This work was jointly supported by
the National Natural Science Foundation of China (No. 61333002, 61503026, 61773053),  the Fundamental Research Funds for the China Central Universities of
USTB (No. FRF-TP-16-024A1, FPR-BD-16-005A, FRF-GF-17-A4),  the Beijing Key Discipline Development Program (No. XK100080537), and the Beijing Natural Science Foundation (No. 4182039).

\bibliographystyle{IEEEtran_doi}
\bibliography{Scheduling}

\end{document}